\documentclass[superscriptaddress,prx,twocolumn,amsmath,amssymb,nofootinbib]{revtex4-1}
\usepackage[T1]{fontenc}
\usepackage{amsthm}
\usepackage{graphicx}
\usepackage{times}
\usepackage{hyperref}
\hypersetup{
colorlinks=true,
linkcolor=red,
citecolor=red,}
\usepackage{textcomp}
\usepackage{enumerate}
\usepackage{multirow}
\usepackage{tabularx}
\usepackage[mathscr]{euscript}
\usepackage{mathtools}
\usepackage{mhchem}
\usepackage[mathscr]{euscript}
\usepackage{dsfont}
\usepackage{amsfonts}
\usepackage[normalem]{ulem}
\usepackage{centernot}
\usepackage{stmaryrd}
\usepackage{longtable}
\usepackage{enumitem,kantlipsum}
\usepackage{array}
\usepackage{wrapfig}
\usepackage{tabu}
\usepackage{graphicx}
\usepackage{placeins}
\usepackage{enumitem}
\usepackage{dutchcal}
\usepackage{tabularx}
\usepackage[export]{adjustbox}
\usepackage{color,soul}
\usepackage{longtable}

\newtheorem{lemma}{Lemma}
\newtheorem{theorem}{Theorem}

\newcommand\s[1]{_{\rm #1}}
\newcommand\sn[2]{_{{\rm #1}:#2}}
\newcommand\us[1]{^{\rm #1}}
\newcommand{\syn}{{\rm{s}}}
\newcommand{\Bell}{{\rm{Bell}}}

\newcommand{\bra}[1] {\langle #1 |}
\newcommand{\ket}[1] {| #1 \rangle}
\newcommand{\braket}[2] {\langle #1 | #2 \rangle}
\newcommand{\ketbra}[1]{ | #1 \rangle\!\langle #1 |}

\newcommand{\ii} {\textbf{i}}

\newcommand{\Tr} {\operatorname{Tr}}

\newcommand{\expec}[1]{\left\langle #1 \right\rangle}
\newcommand{\one}{\leavevmode\hbox{\small1\normalsize\kern-.33em1}}

\newcommand{\oprod}[1] {^{\otimes #1}}

\newcommand{\ind}{\mathbf{I}}

\newcommand{\LL}{{\cal Lin}}

\newcommand{\PP}{{\cal P}}

\newcommand{\CC}{{\cal C}}

\newcommand{\EE}{{\cal E}}
\newcommand{\LE}{{\mathscr LE}}
\newcommand{\St}{{\cal St}}

\newcommand{\B}{{\cal B}}

\newcommand{\sep}{{\cal Sep}}

\newcommand{\MM}{\mathsf{M}}
\newcommand{\pp}{\mathsf{P}}

 \def\ii{\mathord{\rm i}}
 \def\mod{\mathord{\rm mod}}

\def\tr{\mathop{\rm Tr}}

\renewcommand{\ii}{{\rm i}}

\def\beq{\begin{equation}}
\def\eeq{\end{equation}}

\def\barray{\begin{eqnarray}}
\def\earray{\end{eqnarray}}

\newcommand{\FS}[1]{\textcolor{red}{#1}}
\newcommand{\ARB}[1]{\textcolor{blue}{#1}}

\begin{document}
\title{Efficient and robust certification of genuine multipartite entanglement \\
in noisy quantum error correction circuits
}

\author{Andrea Rodriguez-Blanco}.
\email{Electronic address: andrer22@ucm.es}
\affiliation{Departamento de F\'{i}sica Teorica, Universidad Complutense, 28040 Madrid, Spain}
\author{Alejandro Berm{u}dez}
\affiliation{Departamento de F\'{i}sica Teorica, Universidad Complutense, 28040 Madrid, Spain}
\author{Markus M\"{u}ller}
\affiliation{Department of Physics, Swansea University, Singleton Park, Swansea SA2 8PP, United Kingdom}
\affiliation{Institute for Quantum Information, RWTH Aachen University, D-52056 Aachen, Germany}
\affiliation{Peter Gr\"{u}nberg Institute, Theoretical Nanoelectronics,
Forschungszentrum J\"{u}lich, D-52425 J\"{u}lich, Germany}
\author{Farid Shahandeh}
\affiliation{Department of Physics, Swansea University, Singleton Park, Swansea SA2 8PP, United Kingdom}
 

\begin{abstract}
Ensuring the correct functioning of quantum error correction (QEC) circuits is crucial to achieve fault tolerance in realistic quantum processors subjected to noise.
The first checkpoint for a fully operational QEC circuit is to create genuine multipartite entanglement across all subsystems of physical qubits.
We introduce a conditional witnessing technique to certify genuine multipartite entanglement (GME) 
that is efficient in the number of subsystems and, importantly, robust against experimental noise and imperfections.
Specifically, we prove that the detection of entanglement in a linear number of bipartitions 
by a number of measurements that also scales linearly, suffices to certify GME.
Moreover, our method goes beyond the standard procedure of separating the state from the convex hull of biseparable states, yielding an improved finesse and robustness compared to previous techniques.
We apply our method to the noisy readout of stabilizer operators of the distance-three topological color code and its flag-based fault-tolerant version. 
In particular, we subject the circuits to combinations of three types of noise, namely, uniform depolarizing noise, two-qubit gate depolarizing noise, and bit-flip measurement noise.
We numerically compare our method with the standard, yet generally inefficient, fidelity test and to a pair of efficient witnesses, verifying the increased robustness of our method.
Last but not least, we provide the full translation of our analysis to a trapped-ion native gate set that makes it suitable for experimental applications.
\end{abstract}

\maketitle

\section{Introduction}\label{sec:intro_text}

The continued effort to control quantum systems with ever increasing accuracy has led to various
quantum processors ~\cite{Ladd2010}, comprising tens of qubits and, more recently, to the demonstration of quantum supremacy~\cite{Arute2019}. 
However, in the mid- and long-term, an important goal towards real-world applications of quantum computation is to demonstrate
fault-tolerance (FT) with low resources~\cite{gottesman2016quantum,PhysRevX.7.041061}.  
Quantum error correction (QEC)~\cite{PhysRevA.52.R2493,PhysRevA.54.1098,PhysRevLett.77.793} is a building block of the FT theory~\cite{nielsen00,10.5555/1209345} 
the purpose of which is to protect the encoded information from the detrimental accumulation of errors by detecting and correcting them along with the computation.
A current goal in this direction is to reach the break-even point where FT-QEC demonstrates advantage over computations with bare physical qubits.
An initial proposal to attest this benchmark in near-term devices was to compare the performance of a family of QEC circuits acting on logical qubits with the best-possible counterparts using unencoded physical qubits~\cite{gottesman2016quantum}. 
We note, however, that any attempt towards the rigorous demonstration of FT must first ensure that QEC circuits function correctly.

While physical qubits and gates can be characterized using different 
techniques, e.g. spectroscopic and interferometric methods extracting the coherence times~\cite{Schmidt_Kaler_2003}, and randomized benchmarking or process tomography characterizing gate performances~\cite{PhysRevLett.97.220407,PhysRevLett.102.090502}, an efficient characterization for medium-sized registers and circuits that implement QEC protocols is still lacking.
Simulating such circuits for general noise and nontrivial-size encondings to demonstrate this break-even point is computationally very demanding, even in the light of the recent result by Beale and Wallman~\cite{2003.10511} wherein the computational cost of effective-logical-noise simulation is significantly reduced.
Alternative milder criteria for a quantitative evaluation of the performance of QEC circuits show that the requirements to demonstrate beneficial QEC are at reach of near-future quantum platforms, e.g.~using trapped ion quantum processors~\cite{PhysRevX.7.041061, Linkee1701074,Trout_2018}. 
Nonetheless, prior to the development of such optimal full codes for QEC, which must function in the vicinity of FT thresholds, it would be desirable to leverage the currently available noisier devices, e.g. by analysing the performance of building blocks of QEC codes.

To achieve this goal, we note that the capacity to generate entanglement lies at the very heart of QEC
protocols, since preparation circuits of encoded logical states~\cite{Nigg302,DiCarlo2010}, and syndrome readout schemes~\cite{Linkee1701074, Corcoles2015, Andersen2020}, generate
genuinely multipartite entangled (GME) states.
On the same note, the necessity of unbounded GME to obtain an advantage in any quantum computation algorithm was noted by Jozsa and Linden~\cite{jozsa_linden_2003}.
As a result, certifying GME serves as a fundamental benchmark in the progress towards FT quantum computation with noisy quantum processors.

In general, verifying GME is a daunting task, which can be understood from the following four points.
Theoretically, GME detection requires certifying entanglement within all possible bipartitions of a multipartite system.
The number of bipartitions, however, grows exponentially with the number of subsystems. This implies that, in general, the required number of bipartite entanglement tests  increases exponentially with the addition of every subsystem.
Furthermore, certifying bipartite entanglement in its own is known to be NP-hard~\cite{10.1145/780542.780545,DBLP:journals/qic/Gharibian10}.
From an experimental perspective, not every theoretically-conceived measurement can be implemented in practice, calling for the compliance of entanglement certification methods with such limitations.
In addition, experiments suffer from 
imperfections that may  
reduce the capabilities of naive theoretical schemes of entanglement detection considerably.
It is thus necessary to devise and continuously improve efficient GME certification methods. 

An elegant approach that addresses these four 
challenges to a reasonable extent is entanglement witnessing~\cite{HORODECKI19961,Terhal2000,Lewenstein2000,Sperling2013,Horodecki2009}.
An entanglement witness is an observable for which a negative expectation value in a given state indicates that the state is necessarily entangled. 
Interestingly, entanglement witnesses can be constructed from local and experimentally-friendly observables~\cite{Toth2005}, and further optimized to show some degree of noise resilience~\cite{Lewenstein2000,Guhne2006,Shahandeh2017UEW}.
In particular, the guaranteed existence of a witness for each entangled state makes witnesses extremely useful in experiments~\cite{HORODECKI19961}.
Nevertheless, in multipartite scenarios, there seems to be a trade-off between the finesse of entanglement witnesses and the number of measurements needed to carry them out.
Specifically, as will be discussed in detail below,
the standard method of GME certification using a single witness~\cite{Toth2005,Sperling2013}, reduces the exponential growth of the required number of bipartitions and measurements to a linear one by trading the individual bipartitions for the convex hull of all biseparable states, and detecting those states outside this convex hull.
This advantage comes at a price, namely, a significant reduction in the robustness of the witness against experimental imperfections.

In this work, we introduce \textit{conditional entanglement witnessing,} a witnessing technique that combines ideas from localizable entanglement~\cite{Verstraete2004} with entanglement witnessing, in order to reduce the exponential complexity of the GME detection to a linear number of tests while,
at the same time, achieving a remarkable robustness 
against noise and errors. 
Our approach facilitates the detection of GME in quantum states that, although lying within the convex hull of biseparable states, remain masked to the single-witness approach.

We prove that the separation between any quantum state and only a linear number of conditional biseparable sets suffices for the purpose of certifying GME in that particular state. 
Due to the fact that for each individual bipartition the corresponding set of biseparable states is delimited closely by the convex set of conditional biseparable states, our conditional witnesses allow for a finer state discrimination between GME and biseparable states even in presence of noise and experimental imperfections.
We show the efficiency and robustness of conditional witnessing by applying it to the output of non-FT and FT
stabilizer measurement circuits of a $d=3$ topological QEC color code ~\cite{PhysRevLett.97.180501} subjected to phenomenological, circuit-level, and measurement noise models. 
We then numerically compare the performance of our method against the standard witnessing methods.
Our results demonstrate that conditional witnessing is simultaneously more efficient and more robust in proving GME in noisy quantum systems.

The manuscript is organized as follows. 
In Sec.~\ref{sec:FT_colorcode}, we give an overview of the stabilizer formalism for the topological QEC color code, together with the flag-based FT QEC scheme. 
Section~\ref{sec:ideal_circuits} discusses ideal non-FT and FT QEC circuits.
In Sec.~\ref{sec:cond_witnessing_sec}, multipartite entanglement and its standard witnessing are reviewed.
We then introduce our method of conditional entanglement witnessing. 
Section~\ref{sec:noisy_trappedions} is devoted to the compilation of the QEC circuits into the native trapped-ion gates.
Next, phenomenological, circuit-level and measurement error models used in the simulations are introduced.
In Sec.~\ref{sec:robustnes_GME_trappedions} we present our numerical results for our conditional witnessing technique 
and some of the best standard witnessing methods available applied to noisy trapped-ion platforms. 
We discuss the robustness of the methods based on results obtained for different noise channels applied to non-FT and flag-based FT trapped ion circuits.
Finally, in Sec.~\ref{sec:conclusions} we provide the conclusions and outlook of our work.

\section{Flag-qubit-based protocol for fault-tolerant QEC}\label{sec:FT_colorcode}

Quantum error correction (QEC) promises to battle environmental decoherence and experimental imperfections during a quantum computation by redundantly encoding quantum information in multipartite logical states~\cite{RevModPhys.87.307}. 
In order to protect the information that is stored redundantly, active QEC considers performing certain type of measurements frequently 
to detect and correct the error that has corrupted the state without affecting the encoded information~\cite{PhysRevA.52.R2493,PhysRevA.54.1098,PhysRevLett.77.793}.
Many leading QEC codes can be described within the stabilizer formalism~\cite{stabilisers}.

An $[[n,k,d]]$ stabilizer quantum code which
encodes $k$ logical qubits into $n$ physical qubits with a code distance $d$ can correct up to $t=\lfloor{(d-1)/2}\rfloor$ errors. 
The stabilizers $S_{i}$ form an Abelian subgroup $\mathcal{S}$ of size $|\mathcal{S}|=2^{G}$
of  the $n$-qubit Pauli group $\mathcal{P}_{n}=\{\pm 1, \pm i\}\times\{I,X,Y,Z\}^{\otimes n}$ where $I,X,Y,Z$ are the single-qubit Pauli matrices. 
Such a subgroup is typically specified via a generating subset 
of $G$ linearly-independent and mutually commuting Pauli operators, $\mathcal{g}:=\{g_{i}\}_{i=1}^G$, that are known as 
\textit{parity checks}---or subgroup generators---so that $\mathcal{S}=\langle g_{1},...,g_{G}\rangle$. 
The latter means that any element of the stabilizer subgroup $S\in \mathcal{S}$ can be obtained
as a specific product of the generators
of the form
\begin{equation}
    S=\bigotimes_{i} g_{i} \quad \text{with}\quad g_i\in\mathcal{g}.
\end{equation}
The  measurement of a single stabilizer yields one of its two possible eigenvalues $\lambda=\pm 1$ which is $2^{n-1}$-fold degenerate.
Since the parity checks form a set of compatible observables, it is possible to unambiguously specify a $2^{n-G}$-dimensional subspace of the $n$-qubit state space through their common eigenspace with eigenvalue $+1$, 
\begin{equation}
   \mathcal{L}= \{\ket{\psi}\in\mathbb{C}^{2n}:\hspace{2ex}S\ket{\psi}=\ket{\psi},\hspace{1ex} \forall S\in\mathcal{S}\}.
\end{equation}

The subspace $\mathcal{L}$ is known as the \textit{code space} which can be used to embed (codify) the $k=n-G$ logical qubits.
The logical generators
$\{X^{L}_\ell,Z^L_\ell\}_{\ell=1}^k$ are then elements of the Pauli group $\PP_n$ that commute with all stabilizers $S\in\mathcal{S}$, but which are not stabilizers themselves.
Defining the \textit{centralizer} of $\mathcal{S}$ in $\PP_n$ as $\CC({\mathcal{S}})=\{P\in\PP_n:PS=SP,~\forall S\in\mathcal{S}\}$, logical operators lie within $\CC({\mathcal{S}})\setminus\mathcal{S}$.
The key idea here is that the logical operators transform the logical state within the code space while leaving the code space as a whole invariant.

It is central to the construction of a QEC code using the stabilizer formalism that the measurements of parity checks allow us to extract the so-called error syndrome without altering the stored quantum information. 
This typically requires the use of ancillary qubits into which the error syndrome is mapped through a sequence of entangling gates followed by a projective measurement on them. 
In practice, however, the syndrome measurements will also be affected by decoherence and noise, thus introducing further errors. 
This raises a crucial aspect of QEC, namely, the notion of {\it fault-tolerance} (FT)~\cite{FTQEC}---the circuits used for the syndrome readout and other operations must be designed in such a way that errors do not proliferate in an uncontrolled manner, cascading into multiple errors that would affect 
larger parts of the quantum register.
In other words, the QEC code should cope with errors stemming from both the quantum computation and the syndrome extraction even if the measurement process is faulty.
Such FT circuit constructions are essential to harness the full potential correcting power of the QEC code under consideration.

 \begin{figure}[t!]
  \includegraphics[width=1\columnwidth]{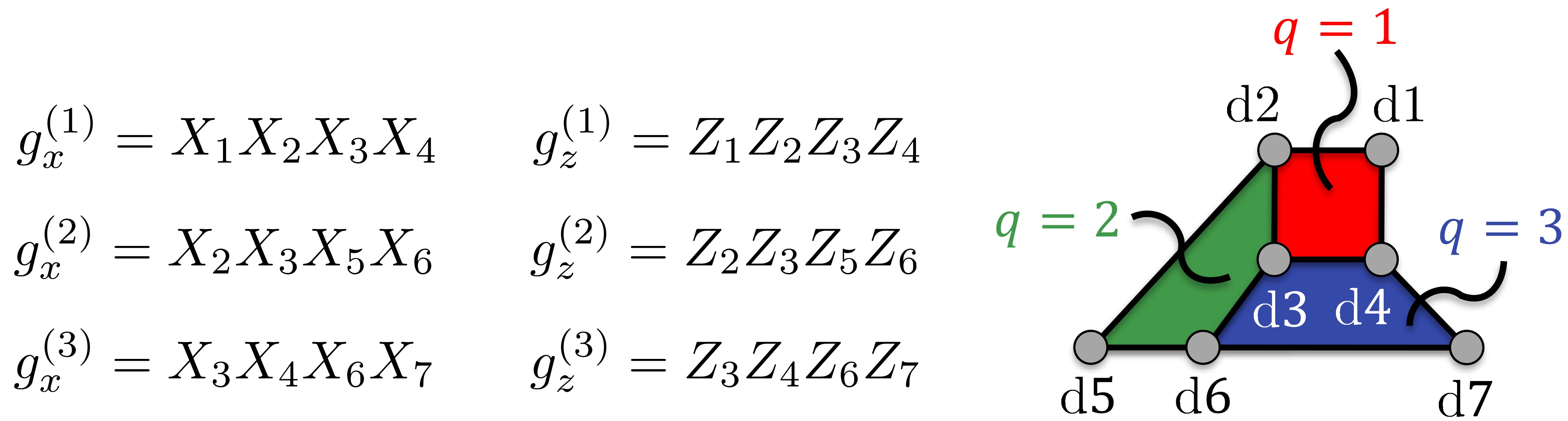}
  \caption{The 7-qubit color code arranged in a three-colorable planar lattice with qubits lying on the vertices, and two types of parity checks per plaquette, hence $G=6$. This code yields a single logical qubit, $k=1$, encoded in $n=7$ data qubits, distance $d=3$, and can correct for $t=1$ error.}
  \label{fig:color_code}
\end{figure}

One of the most attractive strategies for active QEC is based on topological codes where the physical qubits are arranged on planar lattices and 
the stabilizers have a local support, i.e.~they only involve neighboring qubits in the planar lattice~\cite{RevModPhys.87.307}. 
In contrast to QEC codes based on concatenation~\cite{PhysRevLett.77.793}, the syndrome extraction for topological codes only requires local measurements 
simplifying in this way the experimental implementation of rounds of QEC, i.e. syndrome extraction followed by the correction of the most-likely error. 
For instance, for topological color codes~\cite{PhysRevLett.97.180501}, physical qubits are arranged on the vertices of a trivalent three-colorable planar lattice (see Fig.~\ref{fig:color_code}). 
The parity check operators are the following pair of operators ($X$- and $Z$-type parity checks) per plaquette $q$ of the lattice, and have a local support on the qubits located at the vertices of such plaquette
\begin{equation}\label{eq:plaquette_gen}
    g_x^{(q)}=\bigotimes_{i\in v(q)}X_i,\qquad g_z^{(q)}=\bigotimes_{i\in v(q)}Z_i,
\end{equation}
where $v(q)$ is the set of vertices that belong to plaquette $q$.
Accordingly, 
there are $G=6$ parity checks for the smallest color code with $n=7$ physical qubits each acting locally
on only 4 qubits so that the code space is two-dimensional encoding $k=n-G=1$ logical qubit.

To detect the errors that may have corrupted the logical qubit, one must 
measure all 4-qubit parity checks using a FT scheme and infer the location of the most likely error from the classical information obtained from the sequence of $\pm 1$ parity check outcomes known as the error syndrome. 
In this particular example, the parity-check information is 
extracted via projective measurements on an ancillary syndrome qubit.
In Fig.~\ref{fig:stabiliser_readout}.~a, we show a circuit for the readout of an $X$-type parity check of a single plaquette of the 7-qubit color code. 
We note, however, that since all four plaquette qubits are coupled to the same ancillary qubit through the sequential two-qubit gates, single-qubit errors can proliferate and create errors with a larger support. Following the syndrome measurement approach using a single ancilla, any 
such error will result in a logical error, thereby violating the FT design principle.

A recent proposal for an efficient recovery of  FT during  the  syndrome extraction 
is the {\it flag-based} stabilizer readout~\cite{PhysRevLett.121.050502,Chamberland2018flagfaulttolerant,Chao2020}.
As shown in Fig.~\ref{fig:stabiliser_readout}~b, this technique uses a single additional ancillary qubit, the \textit{flag} qubit, which is coupled to the syndrome qubit.
We note that, in general, even if more flag qubits are required for larger-distance codes~\cite{Chamberland2018flagfaulttolerant}, these are not required to be initially prepared in a multipartite entangled Greenberger-Horne-Zeilinger (GHZ)~\cite{Greenberger1989} state which must also be verified to avoid a non-FT proliferation of errors~\cite{PhysRevLett.77.3260, PhysRevLett.98.020501}.
This brings about significant advantage compared to cat-state ancillary methods~\cite{PhysRevLett.77.3260, PhysRevLett.98.020501}.
In the case of the 7-qubit color code, as shown in Fig.~\ref{fig:stabiliser_readout}~b, the flag qubit is coupled to the syndrome qubit to detect whether or not multiple errors might have cascaded onto the data qubits. We note that that the flag qubit by itself does not suffice to correct for such higher-weight errors. 
However, when combined with the subsequent parity check measurements using the syndrome qubit, it can be used to unequivocally identify and correct these dangerous errors achieving the desired fault tolerance~\cite{PhysRevLett.121.050502,PhysRevA.100.062307}.
  \begin{figure}[t!]
  \includegraphics[width=1\columnwidth]{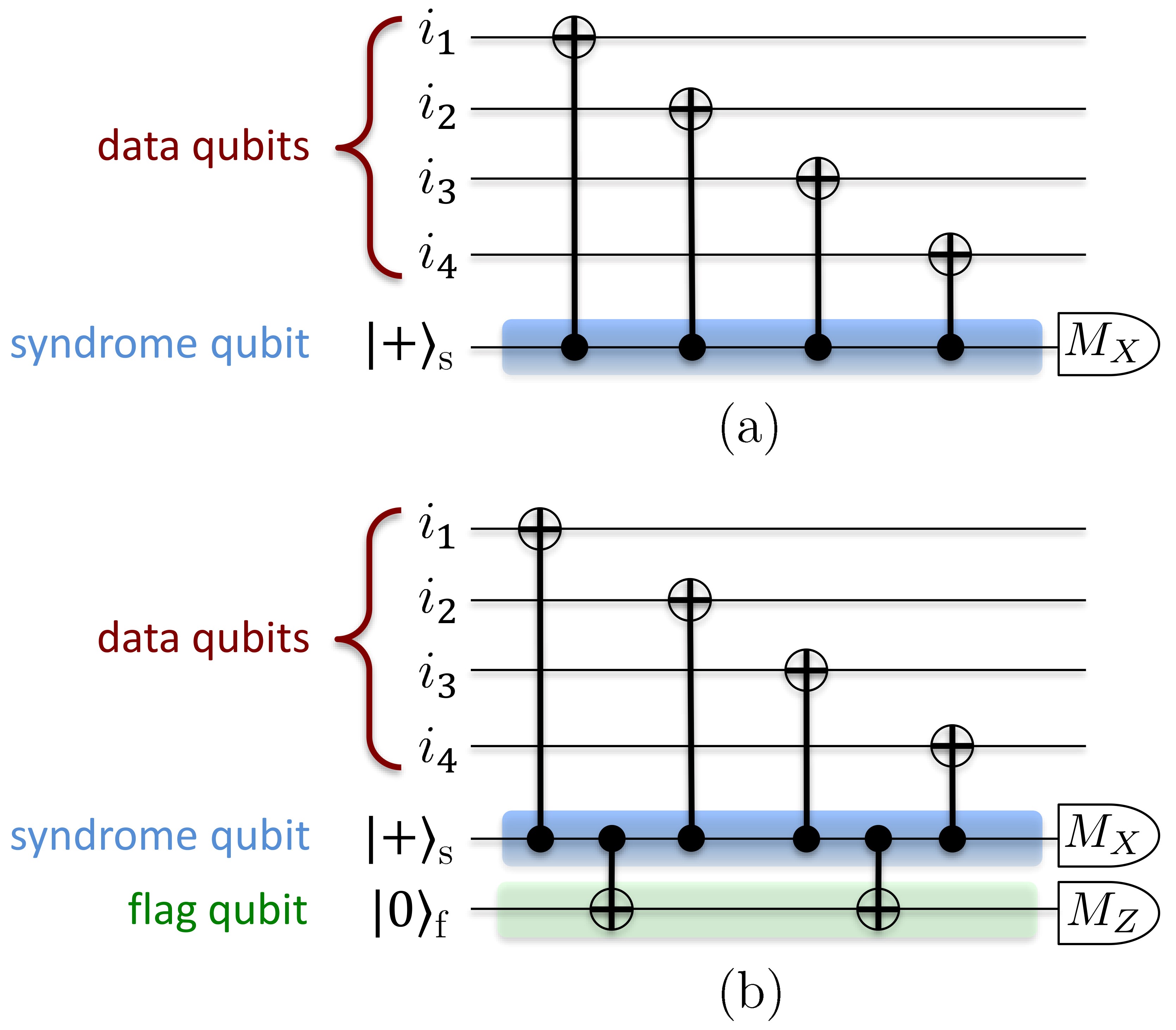}
  \caption{Error syndrome extraction circuits in the color code for (a) the non-FT and (b) the flag-based FT readouts of the 
  $X$-type parity check $g_x^{(q)}$ in Eq.~\eqref{eq:plaquette_gen} of any plaquette $q$ in Fig.~\ref{fig:color_code}.
  The physical/data qubits correspond to the corresponding four physical qubits of this plaquette with indexes $i_1,i_2,i_3,i_4\in\{\rm d1,d2,d3,d4,d5,d6,d7\}.$ These qubits are sequentially coupled via CNOT entangling gates to the syndrome qubit, which must be initialized in a specific state, and then measured in the corresponding basis. In the flag-based scheme, an additional ancillary qubit is appropriately initialized, coupled, and measured
  to reveal crucial information to attain FT.}
  \label{fig:stabiliser_readout}
\end{figure}

\section{Ideal Plaquette Circuit}\label{sec:ideal_circuits}
We are interested in evaluating the correct functioning of the plaquette readout circuits, as they are the smallest building blocks of the circuitry for the topological color code~\cite{PhysRevLett.97.180501}.
To do that, it is useful to first work out the details of the ideal circuits shown in Fig.~\ref{fig:stabiliser_readout}.
Here, by ideal we mean circuits that are free from state preparation and measurement (SPAM) and gate errors.
This is done in the following section.
We emphasize that, here, we only focus on the $g_x$ generator.
Nevertheless, a similar procedure can be implemented using the $g_z$ generator as presented in Appendix~\ref{app:gz_ideal}.

\subsection{Ideal non-FT plaquette circuit}\label{sec:nonFTcircuit}
Let us start our analysis with the simplest example of the ideal plaquette circuit  in Fig.~\ref{fig:stabiliser_readout}.~a. 
In a full QEC code, the syndrome qubit can be prepared in $\ket{+}\s{s}=(\ket{0}\s{s}+\ket{1}\s{s})/\sqrt{2}$, while
the data qubits $\{i_{1},i_{2},i_{3},i_{4}\}$ are, in general, part of a larger logical state.
After applying the circuit, the syndrome qubit would contain the $Z$-type error information.
This information could then be extracted via a projective measurement of the syndrome qubit. 
In order to evaluate the performance of the color code blocks, however, we do not need to read out the syndrome.
In this case, our interest lies in the characterization of the corresponding circuit from the perspective of GME. 
To do so, it turns out to be sufficient to assume that data qubits are initialized 
in $\ket{0}\oprod{4}$.
Next, four CNOT gates are sequentially applied with the syndrome as the common control qubit while the data qubits play the role of the corresponding targets. 
Recall that a CNOT gate between the control and the target qubits, $\rm c$ and $\rm t$, respectively, can be written in terms of Pauli operators as $\mathrm{CNOT}_{c,t}={1 \over 2}\{ (I_{c}+Z_{c})I_{t}+(I_{c}-Z_{c})X_{t}\}$.
After some algebra, the output state obtained is 
\begin{equation}\label{eq:output_5qubit}
    \ket{\psi\s{out}}={\ket{0}\oprod{5} + \ket{1}\oprod{5} \over \sqrt{2}}=\ket{\mathrm{GHZ}\oprod{5}_{+}}.
\end{equation}

The resulting $5$-qubit GHZ state of Eq.~\eqref{eq:output_5qubit} can be understood as the one-dimensional code space of the stabilizer subgroup 
\begin{equation}\label{eq:5GHZ_stabilizergroup}
\begin{split}
     \mathcal{S}^{\rm GHZ}\s{5q} = \langle & g_1 = Z_1Z_2, g_2=Z_2Z_3, g_3=Z_3Z_4,\\
     &\quad  g_4=Z_4Z_5, g_5=X_1X_2X_3X_4X_5\rangle,
\end{split}
\end{equation}
and thus, it can be recast as the product of the projections onto the common +1 eigenspace of $\mathcal{S}\s{5q}$, namely
\begin{equation}\label{eq:GHZ5Stabilizers}
    \varrho\s{out}=\ket{\mathrm{GHZ}\oprod{5}_{+}}\bra{\mathrm{GHZ}\oprod{5}_{+}}=\prod_{i=1}^{5}\left(\frac{I_i+g_{i}}{2}\right)=\frac{1}{2^{5}}\sum_{i=1}^{2^{5}}S_{i}.
\end{equation}
Let us now briefly analyze the effect of a dangerous error in the syndrome qubit.
Consider, for example, the case wherein the syndrome suffers from a Pauli-$X$ error between the second and third CNOT gates; see Fig.~\ref{fig:ideal_circuits}.~a. Using the standard propagation of errors across CNOT gates~\cite{nielsen00},
it is easy to calculate the output state as
\begin{equation}\label{eq:GHZ5Xerror}
    \begin{split}
          \ket{\psi\s{out}} = X\s{s}X_3X_4 \ket{\mathrm{GHZ}\oprod{5}_{+}} = X_1X_2 \ket{\mathrm{GHZ}\oprod{5}_{+}}.
    \end{split}
\end{equation}
It is evident from Eq.~\eqref{eq:GHZ5Xerror} that the syndrome qubit carries no information about the error and that the single-qubit error has turned into a weight-2 error showcasing the non-FT aspect of this circuit. Importantly, even in the case of a coherent rotation error $\exp{(i\phi X)}$ of the syndrome qubit, as depicted in Fig.~\ref{fig:ideal_circuits}.~b, a simple calculation results in the output state
\begin{equation}
    \begin{split}
           & \ket{\psi\s{out}} = \cos{\phi}\ket{\mathrm{GHZ}\oprod{5}_{+}} + \ii\sin{\phi} X\s{s}X_3X_4 \ket{\mathrm{GHZ}\oprod{5}_{+}},
    \end{split}
\end{equation}
which again represents the fact that a projective $X$ measurement of the syndrome qubit does not reveal the error. 
This brings us to the FT scheme of the next section.
\begin{figure}[t!]
  \includegraphics[width=0.7\columnwidth,center]{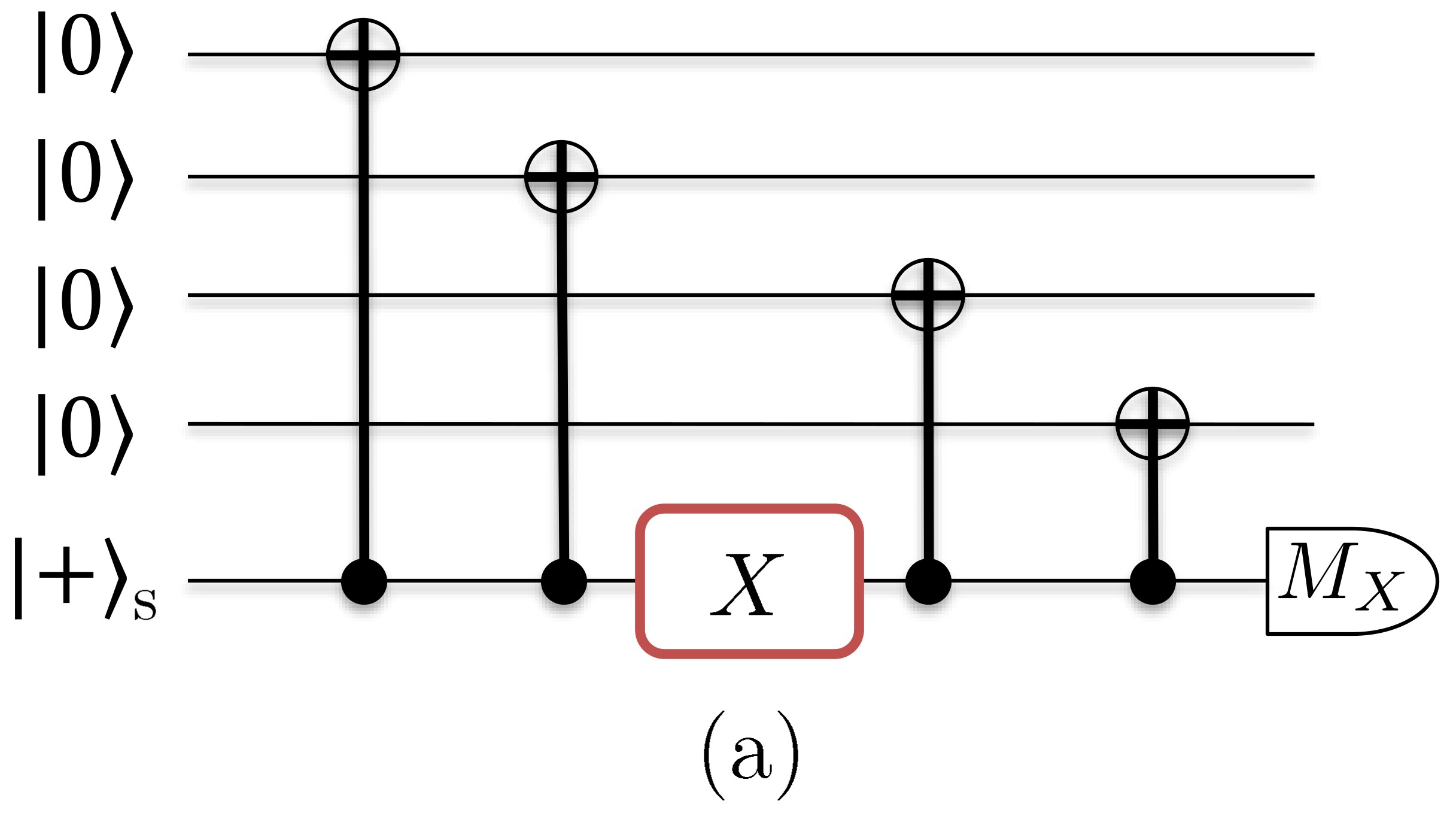}
  \includegraphics[width=0.7\columnwidth,center]{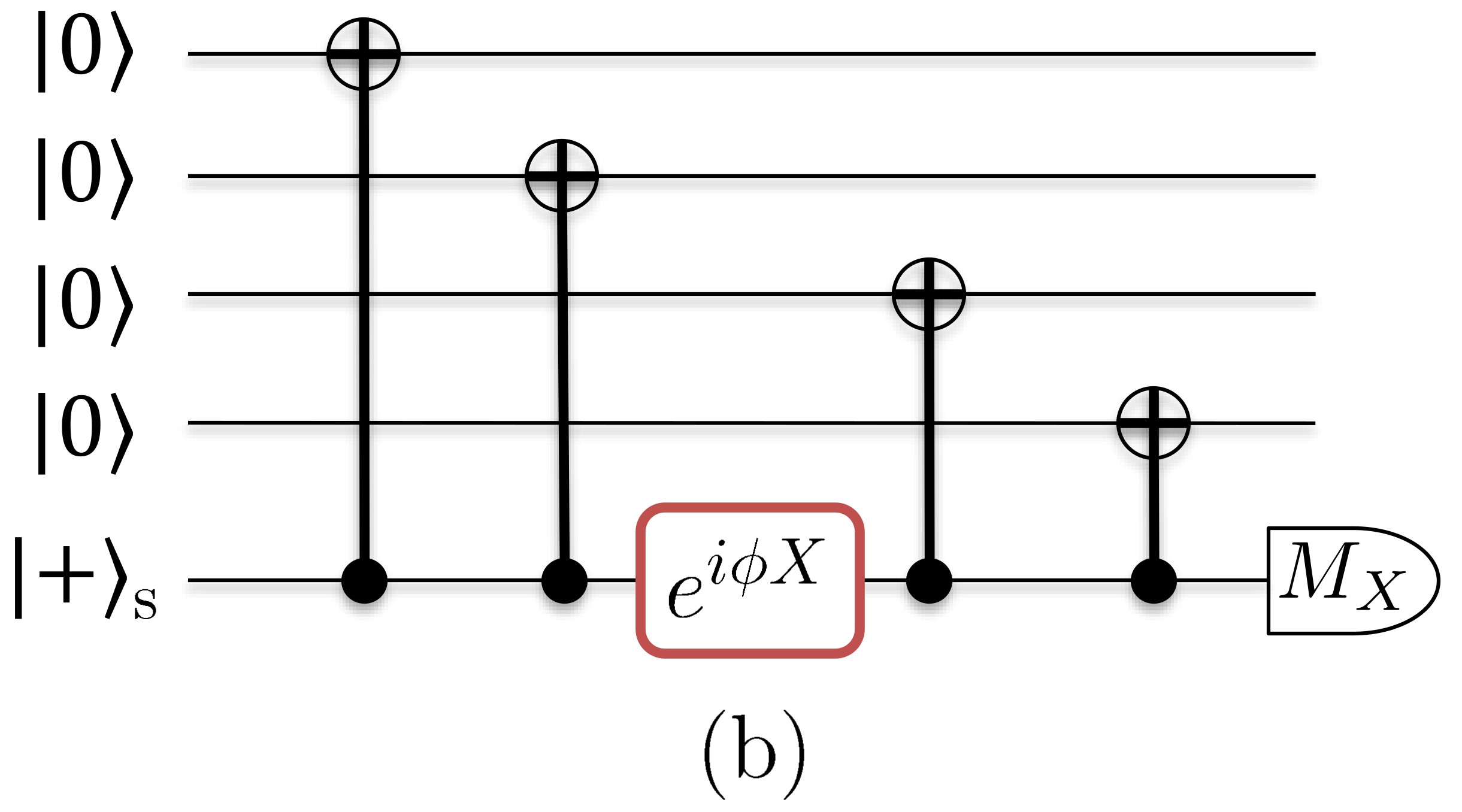}
  \includegraphics[width=0.7\columnwidth,center]{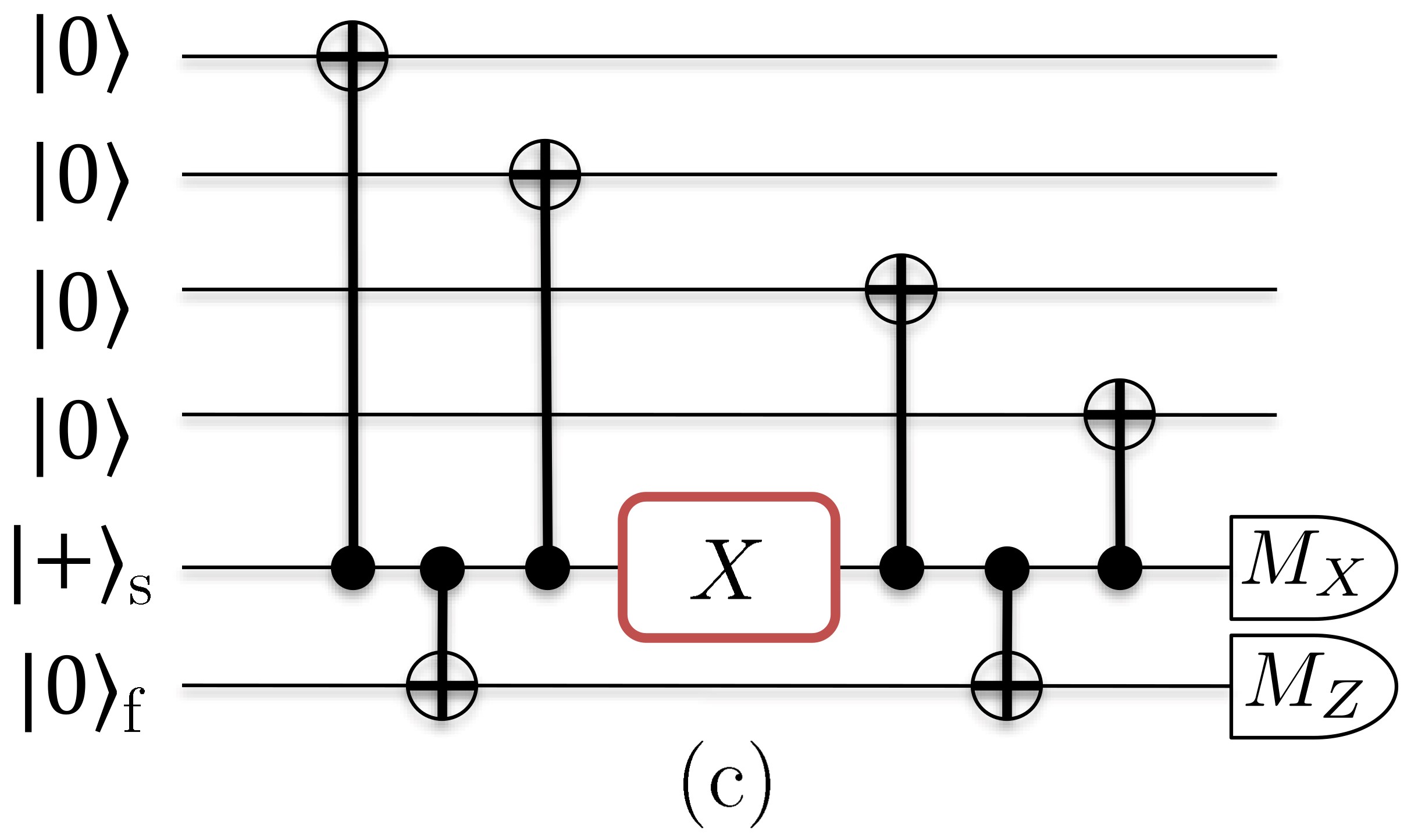}
  \includegraphics[width=0.7\columnwidth,center]{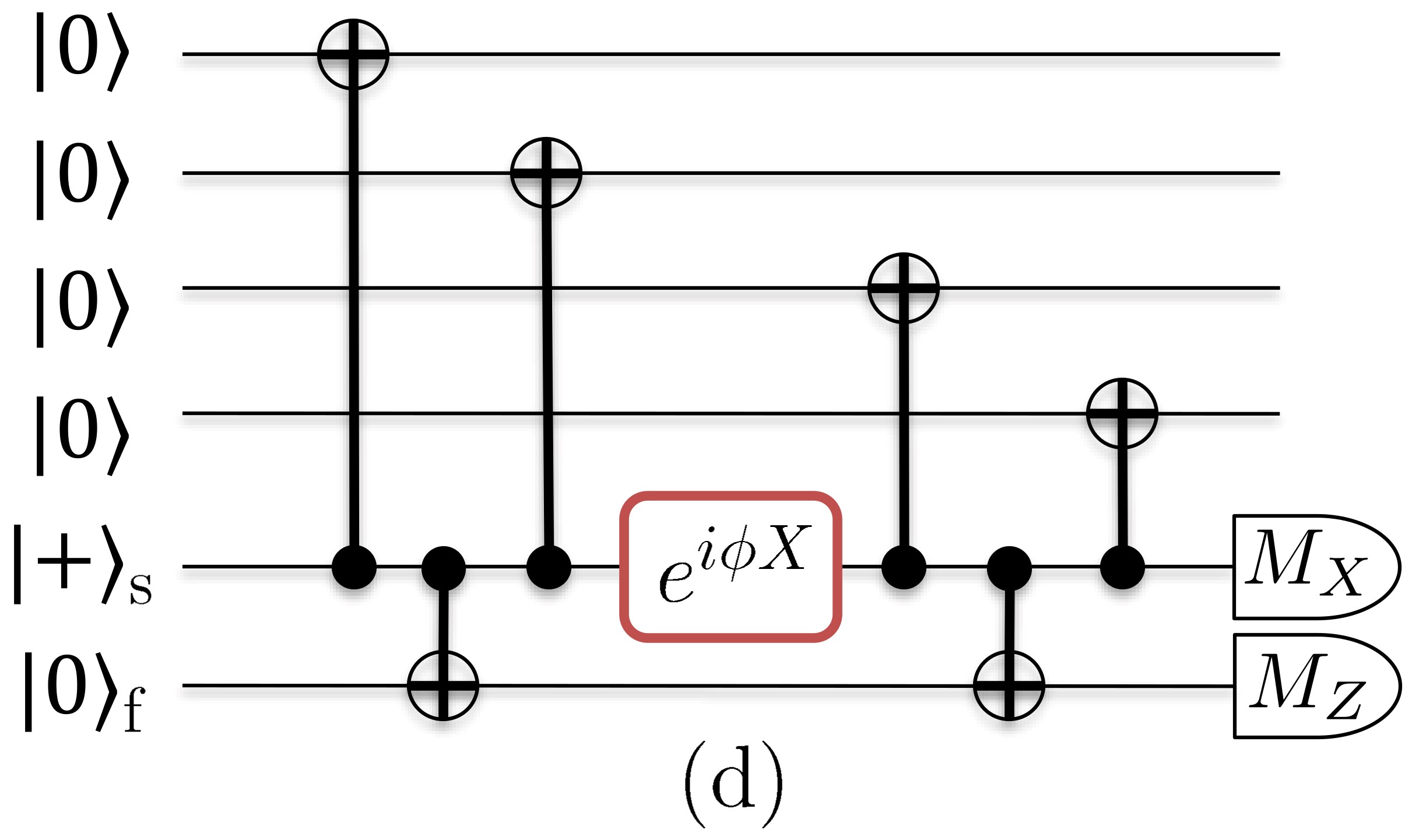}
  \caption{Introduction of dangerous errors in noisy plaquette circuits for (a,~b) the non-FT and (c,~d) the flag-based FT readouts of the $g_x$ generator. 
  (a) A single Pauli-$X$ error occurring on the syndrome qubit in a dangerous position leads to a cascade of errors in the data qubits.
  (b) Coherent-rotations errors in the syndrome qubit also lead to multiple undetectable errors in the data qubits.
  In (c), in contrast to the non-FT cases, outcomes of the measurement on the flag qubit reveals the cascade of errors from the syndrome qubit to the data.
  Similarly in (d), the coherent-rotation error is detected through measurement outcomes of the flag qubit.}
  \label{fig:ideal_circuits}
\end{figure}
\subsection{Ideal flag-based plaquette circuit}\label{sec:FTcircuit}
We now consider the FT flag-based circuit presented in Fig~\ref{fig:stabiliser_readout}.~b. 
Similar to the noiseless non-FT plaquette circuit of Sec.~\ref{sec:nonFTcircuit}, we assume the syndrome qubit is initialized in the state $\ket{+}_{\rm{s}}$ and the data qubits $\{i_{1},i_{2},i_{3},i_{4}\}$ in the product state $\ket{0}\oprod{4}$.
The flag qubit, on the other hand, is prepared in $\ket{0}_{\rm{f}}$.
The main difference between the flag-based and the non-FT circuits is that the two additional CNOT gates between the syndrome and the flag 
propagate the potential error in the syndrome qubit to the flag qubit while this does not corrupt the resulting state in the event of an ideal error-free realization.
This enables one to infer about the error in the syndrome qubit that propagates to multiple errors on the data by a projective measurement of the flag qubit.
When there are no errors in the syndrome, it is expected from the ideal circuit to yield the output state 
\begin{equation}
    \ket{\psi\s{out}}={(\ket{0}\oprod{5} + \ket{1}\oprod{5})\ket{0}\s{f} \over \sqrt{2}}=\ket{\mathrm{GHZ}\oprod{5}_{+}}\ket{0}\s{f}.
    \label{eq:output_stabilizers}
\end{equation}
However, if the syndrome qubit suffers from a single Pauli-$X$ error happening between the two CNOTs, as shown in Fig.~\ref{fig:ideal_circuits}.~c, it will propagate to the flag qubit as a bit flip.
This is reflected in the output state
\begin{equation}\label{eq:pureout4+2}
\begin{split}
    \ket{\psi\s{out}} = X\s{s}X_3X_4 \ket{\mathrm{GHZ}\oprod{5}_+} \ket{1}\s{f}.
\end{split}
\end{equation}
As a result, in contrast to the non-FT scenario of Sec.~\ref{sec:nonFTcircuit}, the $-1$ measurement outcome of the flag indicates the propagation of a correlated weight-2 error in the syndrome and data. Note that there are other events in a noisy experimental realization of the circuit, for example a measurement error on the flag qubit, which could also trigger the flag. Therefore, to construct one FT QEC cycle, further measurements of the stabilizers would be required to discriminate between these cases~\cite{PhysRevA.100.062307}.

Finally, if the syndrome qubit suffers from the coherent-rotation error $\exp{(i\phi X)}$ between the second and third syndrome-data CNOT gates 
shown in Fig.~\ref{fig:ideal_circuits}.~d, then the output state of the ideal circuit becomes 
\begin{equation}\label{eq:pureout4+2}
\begin{split}
    \ket{\psi\s{out}} = & \cos{\phi}\ket{0}\s{f}\ket{\mathrm{GHZ}\oprod{5}_+}  + \ii\sin{\phi} \ket{1}\s{f} X\s{s}X_3X_4 \ket{\mathrm{GHZ}\oprod{5}_+}.
\end{split}
\end{equation}
Here again, the flag unveils a cascaded error in the syndrome and data qubits.

Having understood the working principles and the ideal outputs of each plaquette in both non-FT and flag-based FT scenarios, we are ready to introduce our main tool for characterizing the performance of these circuits, namely, the conditional entanglement witnessing.

\section{Conditional entanglement witnessing}\label{sec:cond_witnessing_sec}

\subsection{Multipartite entanglement and its witnesses}\label{subsec:witness}

In this section we review the well-established theory of multipartite entanglement witnessing~\cite{Eisert2008,Sperling2013}.
The readers familiar with these details might skip this section and move on to Sec.~\ref{subsec:efficientGME}.
We identify the subsystems of an $n$-partite system using an index set $\ind:=\{1,2,\dots,n\}$ so that, given the Hilbert space $\mathcal{H}_i$ associated with each subsystem, the Hilbert space of the total system is given by $\mathcal{H}_\ind=\otimes_{i\in\ind} \mathcal{H}_i$.
Evidently, it is possible to group the subsystems in many different ways.
Calling each possible grouping a \textit{partitioning}, they correspond to partitionings of the index set $\ind$.
An $m$-partitioning $(\ind_1|\ind_2|\cdots|\ind_m)$ is thus a specific grouping of subsystems such that $\cup_x \ind_x=\ind$ and $\ind_x\cap\ind_y = \emptyset$ for any $x\neq y$.
It is customary to call $2$-partitionings \textit{bipartitions}.
We also denote the Hilbert space associated with each party by $\mathcal{H}_{\ind_j}=\otimes_{x\in\ind_j}\mathcal{H}_x$.

Recall that each quantum system is described by a density operator $\varrho$, which is a unit-trace positive member of the set of all bounded linear transformations $\LL(\mathcal{H})$ on the Hilbert space $\mathcal{H}$ assigned to the system. 
The set of all quantum states of a system is denoted by $\St\subset\LL(\mathcal{H})$, which is a compact convex set.
Applying these rules to the multipartite scenario for a presumed $m$-partitioning $(\ind_1|\cdots|\ind_m)$, the state space of each party $\ind_j$ is denoted by $\St_{\ind_j}\subset\LL(\mathcal{H}_{\ind_j})$, while the state space of the total system is identified as $\St_{\ind}\subset\LL(\mathcal{H}_{\ind})$.
Now, imagine that parties in $(\ind_1|\cdots|\ind_m)$ were spatially separated and all they could do was to perform local operations on their possessed quantum systems and possibly \textit{communicate} classical messages, known as the \textit{local operation and classical communication} (LOCC) paradigm, and ask what quantum states of the total system can the parties prepare. 
It can be shown that the quantum states that the parties can achieve in this setting are of the form~\cite{Werner1989}
\begin{equation}\label{eq:ksep}
\varrho=\sum_i p_i \varrho_{\ind_1;i}\otimes\cdots\otimes \varrho_{\ind_m;i},
\end{equation}
wherein $\varrho_{\ind_x;i}\in\St_{\ind_x}$ for any $x\in\{1,\dots,m\}$ and $\{p_i\}$ is  a probability distribution.
The quantum states of this form are called $m$-separable states with respect to the partitioning $(\ind_1|\cdots|\ind_m)$ and their collection forms a closed convex set denoted by $\sep_{\ind_1|\cdots|\ind_m}$.
The set of $m$-separable states for any given partitioning is also dense and convex.
Most importantly, however, is that $\sep_{\ind_1|\cdots|\ind_m}\subsetneq\St_{\ind}$, i.e. there are quantum states of the total system that cannot be prepared by the $m$ parties using LOCC.
Such states are called \textit{entangled with respect to the $m$-partitioning $(\ind_1|\cdots|\ind_m)$}.

Entanglement in an $n$-partite system has a complex structure because the possibilities of choosing the partitioning 
grow exponentially with the number of subsystems~\cite{Sperling2013,Shahandeh2014}. 
A particularly interesting class of multipartite entangled states, however, is the one formed by states that are entangled within all possible bipartitions.
These states are called \textit{genuinely multipartite entangled} (GME), or---sometimes---\textit{fully inseparable}~\cite{Eisert2008}\footnote{Some authors (see e.g. Refs.~\cite{Guhne2009,Huber2010}) define the genuine multipartite entangled states as those not lying within the closed convex hull of all bipartitions.
This means that a mixture of two biseparable states with respect to two different bipartitions is still considered biseparable even if it is not separable within any of the two bipartitions.
We note, however, that changing the entanglement of a quantum state with respect to different bipartitions cannot be changed via LOCC and thus, a mixed state of the mentioned form has nonlocal properties beyond biseparability.}.
Henceforth we use the former name for such states.
This is the strongest form of multipartite entanglement as it can be shown to imply 
entanglement within all possible $k$-partitionings of the system~\cite{Sperling2013,Shahandeh2014}.

Let us remark that  the certification of GME in an $n$-partite system requires verification of entanglement in $N_{\rm b}={1\over 2}\sum_{i=1}^{n-1} \binom{n}{i}=2^{n-1}-1$ different bipartitions, showing the exponential increase in both analytical and practical complexity of this task. Certifying entanglement of an arbitrary quantum state, be it bipartite or multipartite, is also known to be very hard from the perspective of complexity theory~\cite{10.1145/780542.780545,DBLP:journals/qic/Gharibian10}. 
This is because determining whether a given quantum state is entangled or not is equivalent to deciding whether it is inside or outside of a suitable convex set of separable states.
However, any such set is not a simplex, meaning that it does not have a finite  number of ``straight border lines'' to be checked for the separation of the state under consideration~\cite{ZyczkowskiBook}.
Nevertheless, it is possible to obtain many sufficient (but maybe not necessary) conditions for entanglement of quantum states~\cite{Horodecki2009}, a very large class of which are \textit{entanglement witnesses}~\cite{HORODECKI19961,Terhal2000,Lewenstein2000,Horodecki2001}.

Simply speaking, an entanglement witness is a quantum observable whose expectation value is bounded for all separable states.
This bound, however, can be violated by \textit{suitable} entangled states.
More precisely, for every entangled state with respect to a specific $m$-partitioning $({\ind_1|\cdots|\ind_m})$, $\varrho\notin\sep_{\ind_1|\cdots|\ind_m}$, there exists a bounded Hermitian operator $W_{\ind_1|\cdots|\ind_m}$ such that
\begin{equation}\label{eq:GenWitness}
\begin{split}
    \forall \sigma\in\sep_{\ind_1|\ind_2|\cdots|\ind_m} \quad &\Tr \left(W_{\ind_1|\cdots|\ind_m}\sigma\right) \geq 0,\\
    & \Tr \left(W_{\ind_1|\cdots|\ind_m}\varrho\right) <0.
\end{split}
\end{equation}
Importantly, the existence of the witness operator $W_{\ind_1|\cdots|\ind_m}$ is guaranteed by the Hahn-Banach separation theorem~\cite{Horodecki2001}.
Finding a witness operator with properties in Eq.~\eqref{eq:GenWitness} may seem to be difficult.
Thankfully, using the affine properties of the space of bounded linear operators on Hilbert spaces, it was shown~\cite{Lewenstein2000,Horodecki2001,Sperling2009,Sperling2013} that every entanglement witness can be constructed from a bounded linear \textit{test} operator, henceforth denoted by $L$, so that
\begin{equation}\label{eq:GenTest}
    W_{\ind_1|\cdots|\ind_m}=l_{\ind_1|\cdots|\ind_m}I - L_{\ind_1|\cdots|\ind_m},
\end{equation}
in which the so-called \textit{separability bound} is
\begin{equation}\label{eq:GenBound}
    l_{\ind_1|\cdots|\ind_m}=\sup_{\sigma\in\sep_{\ind_1|\cdots|\ind_m}} \left\{\Tr\left( L_{\ind_1|\cdots|\ind_m}\sigma\right)\right\}.
\end{equation}
While the optimization in Eq.~\eqref{eq:GenBound} seems demanding, it is known that entanglement witnessing is efficiently decidable~\cite{Eisert2004,Doherty2005}.
Now, according to Eq.~\eqref{eq:GenWitness}, any state $\varrho$ for which $\Tr \left(L_{\ind_1|\cdots|\ind_m}\varrho\right) > l_{\ind_1|\cdots|\ind_m}$ is
entangled with respect to the partitioning $({\ind_1|\cdots|\ind_m})$.

The program of entanglement witnessing thus proceeds as follows.
It is assumed that some information about the entangled state to be detected is available \textit{a priori},  
e.g. the quantum state  that results from an ideal noiseless circuit is known.
The goal is to devise an experimentally-friendly observable, i.e. one that can be measured with relatively low cost, for instance, in the number of settings,
based on the given information which is also resilient against the state preparation and measurement imperfections.
We note that, in general, different witnesses are required for different partitionings~\cite{Sperling2013,Shahandeh2014}.
In this paper, however, we are interested in GME witnessing and shall thus focus on bipartitions.

Let us give a simple example of how the standard GME witnessing can be performed.
Consider the circuit shown in the top panel of Fig.~\ref{fig:stabiliser_readout}.
We showed in Sec.~\ref{sec:nonFTcircuit} that the ideal output of this circuit is the 5-qubit GHZ state,
\begin{equation}
        \ket{\psi\s{out}}=\ket{\mathrm{GHZ}\oprod{5}_{+}}={\ket{0}\oprod{5}+\ket{1}\oprod{5} \over \sqrt{2}}.
\end{equation}
To show that this state is genuinely multipartite entangled, we have to check the entanglement in all $N_{\rm b}=15$ possible bipartitions.
Let us focus on one of them, say $({\rm s}|1,2,3,4)$ where we used the symbol $\rm s$ for the the syndrome qubit. 
It is easy to see that, with respect to this bipartition, the GHZ state is equivalent to a bipartite Bell state, i.e.,
\begin{equation}
        \ket{\psi\s{out}}=\ket{\mathrm{GHZ}\oprod{5}_{+}}={\ket{0\overline{0}}+\ket{1\overline{1}} \over \sqrt{2}},
\end{equation}
where $\ket{\overline{0}}=\ket{0}^{\otimes 4}$ and $\ket{\overline{1}}=\ket{1}^{\otimes 4}$. 
Consider now the test operator
\begin{equation}\label{eq:stdTestGHZ5}
L={\left(\ket{0\overline{0}}+\ket{1\overline{1}}\right)\left(\bra{0\overline{0}}+\bra{1\overline{1}}\right) \over 2}=\ket{\mathrm{GHZ}\oprod{5}_{+}}\bra{\mathrm{GHZ}\oprod{5}_{+}}.
\end{equation}
Using the similarity with the Bell state it is easy to compute the supremum expectation value of this test operator over all biseparable states $\sep_{\rm{s}|1,2,3,4}$ obtaining the separability bound of Eq.~\eqref{eq:GenBound} with respect to this bipartition as
\begin{equation}\label{eq:gsupGHZ5}
    l_{s|1,2,3,4}=\sup_{\sigma\in\sep_{\rm{s}|1,2,3,4}}\left\{ \tr\left({L\sigma}\right)\right\} = {1\over 2}.
\end{equation}
Consequently, one can define an entanglement witness with respect to this bipartition as
\begin{equation}\label{eq:stdWitnessGHZ5}
    W_{\rm{s}|1,2,3,4}={1 \over 2}I-L,
\end{equation}
so that $\tr\left({W_{\rm{s}|1,2,3,4}\sigma}\right)\geq 0$ for all biseparable states $\sigma\in\sep_{\rm{s}|1,2,3,4}$, while there exist quantum states for which $\tr\left({W_{\rm{s}|1,2,3,4}\varrho}\right)<0$.
Clearly, the 5-qubit GHZ state itself is an example of the latter.
It then follows that $\ket{\mathrm{GHZ}\oprod{5}_{+}}\bra{\mathrm{GHZ}\oprod{5}_{+}} \notin \sep_{\rm{s}|1,2,3,4}$ which must be read as ``the 5-qubit GHZ state is entangled with respect to the bipartition $(\rm{s}|1,2,3,4)$''. For this example, due to the symmetries of the test operator $L$ in Eq.~\eqref{eq:stdTestGHZ5} with respect to different bipartitions, it is easy to show that the maximum expectation value of the test operator for all bipartitions is indeed the same and equal to ${1\over 2}$.
As a result, we find that the witness operator of Eq.~\eqref{eq:stdWitnessGHZ5} is in fact capable of detecting the entanglement of the 5-qubit GHZ state within all possible bipartitions and thus serves as a faithful witness of GME. 
\begin{figure}[t!]
  \includegraphics[width=0.7\columnwidth]{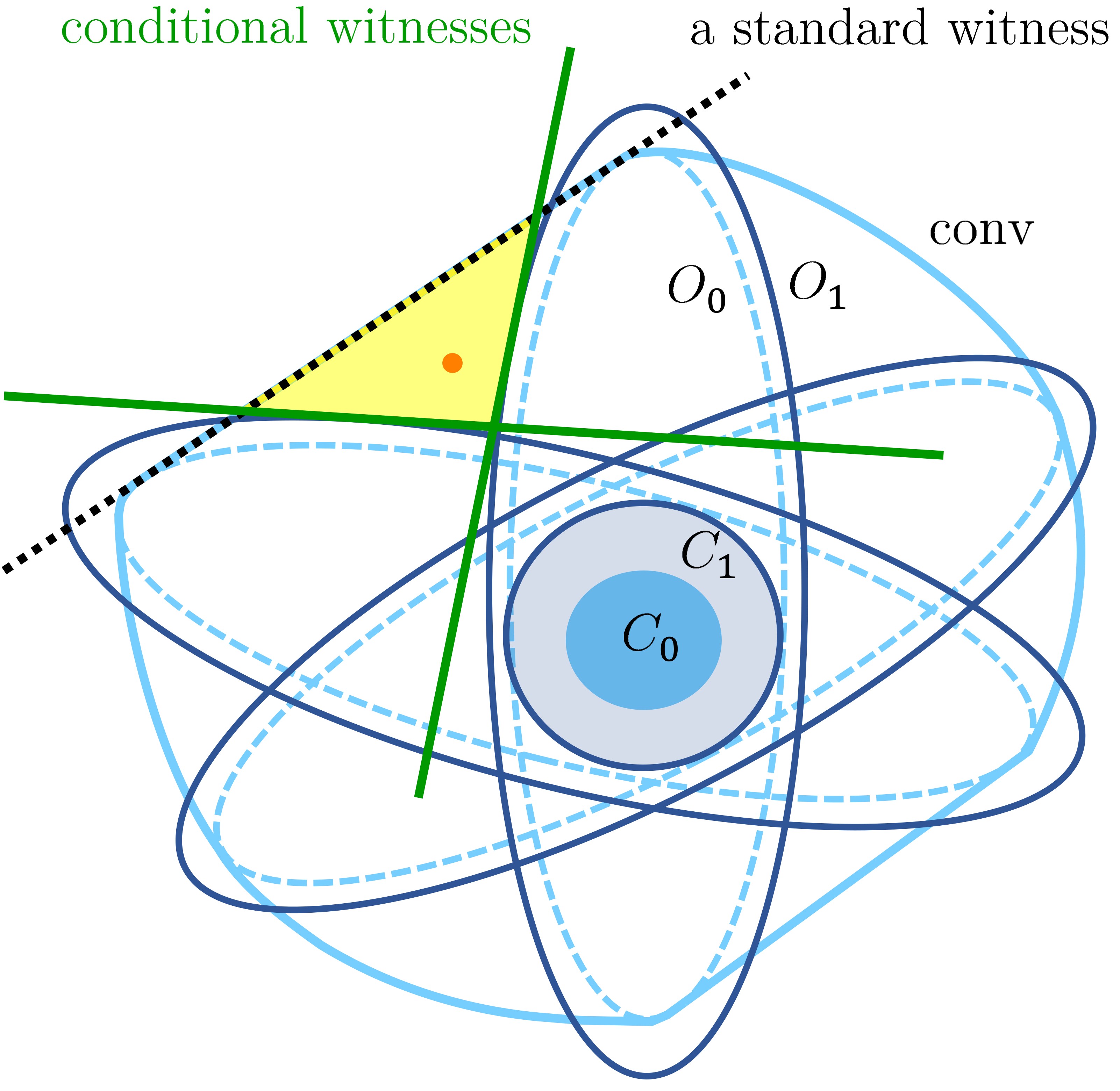}
  \caption{Heuristic comparison between standard GME witnessing and conditional GME witnessing technique of Sec.~\ref{subsec:condentwitness}.
  The light-blue dashed ovals, say $O_0$, represent separable states with respect to a specific individual bipartition, aka biseparable sets.
  Given $n$ subsystems there exist an exponential number, i.e. $N_{\rm b}=2^{n-1}-1$, of biseparable sets. 
  The closed convex hull of all biseparable sets is denoted as ${\rm conv}$. 
  The inner circle $C_0$ is the set of fully separable states.
  The dark-blue ovals, say $O_1$, represent the conditional biseparable states with respect to individual bipartitions, i.e. states for which entanglement cannot be localized between any two subsystems from parties of the bipartition.
  Evidently, $O_0\subsetneq O_1$.
  The circle $C_1$ represents the intersection of all conditional biseparable sets.
  The black dashed line represents a standard GME witness.
  It is clear that there exist quantum states, e.g. the red dot, that are GME and yet lie within the convex hull ${\rm conv}$.
  Using Theorem~\ref{th:measNo} or Theorem~\ref{th:measNoNeighbour} of Sec.~\ref{subsec:efficientGME}, it is enough to certify the entanglement of states with respect to only a linear number, i.e. $n-1$, of conditional bipartitions. 
  We use $n-1$ witnesses (green solid lines) for this purpose. 
  The intersection of detected states using such witnesses are also GME entangled states.  
  Thus, the shaded yellow area highlights the advantage obtained using our conditional GME witnessing as discussed in Sec.~\ref{subsec:condentwitness}.}
  \label{fig:witnesscompare}
\end{figure}

We emphasize, however, that this is not generally the case for arbitrary quantum states, generic test operators, and arbitrary bipartitions. 
In the most general case of certifying entanglement using a single test operator, a given quantum state is GME if the expectation value of the test operator for that state is larger than the maximum of the separability bounds Eq.~\eqref{eq:GenBound} for all bipartitions~\cite{Sperling2013}, i.e.
\begin{equation}\label{eq:GMEmaxBound}
   \varrho \hspace{1ex} {\rm is\hspace{1ex} GME\hspace{1ex} if}\hspace{2ex} \expec{L}={\rm Tr}\left(L\rho\right)>\max_{B\in\B}{l_B},
\end{equation}
where the maximum is taken over the set of all bipartitions $\B$  which has $|\B|=N_{\rm b}$ elements.
Now, suppose that $\varrho$ and $\sigma$ are two biseparable states with respect to two bipartitions $B_1$ and $B_2$, respectively.
Thus, we have $\Tr\left({L\varrho}\right)\leq \max_{B\in\B}{l_B}$ and $\Tr\left({L\sigma}\right)\leq \max_{B\in\B}{l_B}$, which implies $\Tr\left({L(p\varrho+ (1-p)\sigma) }\right)\leq \max_{B\in\B}{l_B}, \forall p\in[0,1]$.
As this holds for any choice of $B_1$ and $B_2$, we conclude that the inequality $\expec{L}>\max_{B\in\B}{l_B}$ can be geometrically interpreted as a
test that serves to identify points that lie outside the closed convex hull of all biseparable states.
In other words, any standard GME witness using a single witness is only capable of detecting GME states that lie outside the closed convex hull of all biseparable states; see Fig.~\ref{fig:witnesscompare}.

Nevertheless, an important aspect of the above construction is its potential experimental feasibility, depending on our choice of the test operator.
First, to obtain the expectation value of the witness, it is enough to only measure the expectation value of the test operator.
Second, the expectation value of the test operator of Eq.~\eqref{eq:stdTestGHZ5} can, in our particular case, be obtained by local measurements of the 32 stabilizers of the 5-qubit GHZ state as given in Eq.~\eqref{eq:GHZ5Stabilizers}.
Hence, it is in practice possible to obtain the witness value and thus test the entangled nature of the state generated by the actual circuit.
We note, however, that the number of stabilizers grows exponentially with the number of qubits which makes this method practically challenging to implement. 
When using a single witness, this problem can cleverly be avoided by trading noise robustness for efficiency in the number of measurements as shown by T\'{o}th and G\"{u}hne in Ref.~\cite{Toth2005}.
We also emphasize that the latter difficulty must be distinguished from the obstacle due to the exponential growth in the number of bipartitions.

\subsection{Efficient certification of GME} \label{subsec:efficientGME}

As we have seen, entanglement witnessing is meant to be an experimentally-friendly technique to certify entanglement.
Interestingly, there exist techniques to push witnessing to its limits so that a given witness delivers an impressive performance even in the presence of severe noise and errors~\cite{Lewenstein2000,Guhne2006,Guhne2007,Shahandeh2017UEW}. 
In multipartite scenarios, however, we face a different problem, namely the scaling of the number of partitions with the input size which seems to be persistent for any method of entanglement certification.

In the scenario of certifying GME that is relevant to our purpose an $n$-partite system contains $N_{\rm b}=2^{n-1}-1$ bipartitions to be tested for entanglement the exponential growth of which is a major challenge for large systems. 
We note that, by choosing the witnessing method for entanglement certification and using the construction of Eqs.~\eqref{eq:gsupGHZ5} and ~\eqref{eq:stdWitnessGHZ5}, one might construct $N_{\rm b}$ witnesses for each bipartition using 
a single test operator $L$.
Nevertheless, in general, there is an exponential number of optimizations as per Eq.~\eqref{eq:GenBound} that must
be performed~\cite{Sperling2013}.
One way to circumvent this practical limitation is to choose a test operator $L$ for which a suitable positive real number $l$
is known---by other means---such that $l$ is larger than the separability bound of Eq.~\eqref{eq:GMEmaxBound} for all bipartitions.
One such technique was adopted by T\'{o}th and G\"{u}hne in Ref.~\cite{Toth2005}.
However, the ease does not come for free:
(i) the obtained witnesses are not the finest ones~\cite{Lewenstein2000}, as there may exist genuinely multipartite entangled states that cannot be detected by any such single witness,
(ii) there is no general recipe for their construction, and 
(iii) they are limited to certain types of witnesses and systems as the value $l$ cannot be deduced for an arbitrary test operator without performing the optimizations.

In the following, we exploit the concept of \textit{localizable entanglement} (LE)~\cite{Verstraete2004,Popp2005} to introduce a general solution to this challenge that collapses the exponential complexity of GME witnessing to a linear growth for all system types in all dimensions.
The price to pay, however, is that our technique only provides a sufficient condition for detecting GME, i.e. not every genuine multipartite entangled state can be detected using our approach. 
We emphasize, however, that for many operational scenarios including measurement-based quantum computation and quantum communication in which LE is a necessary resource~\cite{VandenNest2006,Briegel1998}, our sufficient condition becomes necessary as well, hence providing a necessary and sufficient condition for practical purposes.

Let us begin with a simple example.
Suppose that two parties Alice and Bob are holding quantum systems $\rm A$ and $\rm B$, respectively.
In addition, Alice is holding an ancillary system $\rm C$.
Hence, we effectively have a tripartite system shared between two parties in the bipartition $\rm (AC|B)$.
We denote the quantum state of the total system with respect to this bipartition by $\varrho\s{AC|B}$, and the amount of entanglement shared between Alice and Bob by $\EE(\varrho\s{AC|B})$.
We ask how that is related to the entanglement between subsystems $\rm A$ and $\rm B$.
This is, in general, a difficult question to answer.
Nevertheless, it is clear that if Alice performs a measurement on $\rm C$ and, conditioned on the $i$th outcome of this measurement, she finds out that the corresponding conditional state $\varrho\sn{A|B}{i}$ is entangled, then the original state $\varrho\s{AC|B}$ prior to the measurement must have been entangled with respect to the bipartition $\rm (AC|B)$.
The reason for this conclusion is that we know entanglement cannot increase via LOCC.
If $\varrho\s{AC|B}$ was to be separable with respect to the bipartition $\rm (AC|B)$, then it would be impossible for Alice to create entanglement by a local measurement (i.e. her measurement on $\rm C$) and classically communicating the result (i.e. her $i$th outcome) to Bob.
We thus see that entanglement between $\rm A$ and $\rm B$ conditioned on measurements on $\rm C$ implies entanglement within the bipartition $\rm (AC|B)$. 
Similarly, if we assume that Bob is holding the subsytem $\rm C$, then we arrive at the conclusion that entanglement between $\rm A$ and $\rm B$ conditioned on measurements on $\rm C$ also implies entanglement within the bipartition $\rm (A|BC)$, hence the following lemma~\cite{shahandeh2019assisted}. 
\begin{lemma}\label{lem:condEnt}
Consider a tripartite state $\varrho\s{ABC}$. Given that the conditional state $\varrho\sn{B|A}{i}$ for the $i$th outcome of some measurement on the subsystem $\rm C$ is entangled, this implies that $\varrho\s{ABC}$ is entangled within both bipartitions $\rm (AC|B)$ and $\rm (A|BC)$.
\end{lemma}
The above example is reminiscent of the idea that, in a multipartite system, it might be possible to concentrate entanglement within a smaller set of subsystems by performing measurements on the remaining subsystems---and communicating the outcome---as first considered 
by Verstraete, Popp, and Cirac~\cite{Verstraete2004}.
The maximum amount of entanglement that can be concentrated between Alice and Bob by such local measurements on average is dubbed as the ``localizable entanglement'' (LE).
More precisely, suppose that $\EE$ denotes an entanglement measure of quantum states that cannot be increased via LOCC.
Given $n$ quantum systems in the state $\varrho$, the LE of the bipartition $(s_x|s_y)$ for $s_x,s_y\in\{s_1,\dots,s_n\}$ after local measurements on subsystems $\{s_x,s_y\}\us{c}:=\{s_1,\dots,s_n\}\setminus\{s_x,s_y\}$ with outcomes $\{i\}$ is given by
\begin{equation}\label{eq:LE}
\LE_{x|y}(\varrho)=\sup_{\MM\in\CC} \sum_i p_i \EE(\varrho_{x|y:i}),
\end{equation}
in which $\varrho_{x|y:i}$ is the joint conditional state of subsystems $s_x$ and $s_y$ given the outcome $i$, with probability $p_i$, of the local measurement $\MM$ on subsystems $\{s_x,s_y\}\us{c}$.
Moreover, the supremum is taken over the set $\CC$ of all $(n-2)$-partite local measurements.
For our example, it thus follows that $\LE\s{A|B}>0$ implies entanglement in both bipartitions $\rm (AC|B)$ and $\rm (A|BC)$ as presented in Lemma~\ref{lem:condEnt}.

Our main observation is that the implications of Lemma~\ref{lem:condEnt}, in light of the concept of LE, can in fact be extended to \textit{any} bipartition of \textit{any} number of systems.
Denote a generic bipartition of the $n$-partite system as $(s_1,\dots,s_j|s_{j+1},\dots,s_n)$.
According to Lemma~\ref{lem:condEnt}, entanglement in this bipartition is certified once the entanglement between two subsystems $s_x$ and $s_y$ from each party, i.e. with $x\leq j$ and $y\geq j+1$, {\it conditioned} on appropriate measurement outcomes of the remaining $n-2$ subsystems, is verified.
Naively, given the $j(n-j)$ choices of the pair $s_x$ and $s_y$ for each bipartition
on top of the fact that there are yet $2^{n-1}-1$ bipartitions to be considered, there still seems to be an exponentially-large 
number of entanglement certifications to be performed.
Our central result stated in the following theorem shows that there is a huge redundancy in this process. Leveraging this redundancy allows us to reduce the exponential number of bipartitions to a linear one, hence making the GME detection efficient.
\begin{theorem}\label{th:measNo}
Given an $n$-partite quantum system of any dimension, 
certification of entanglement between the subsystems $s_1$ and $s_x$ for all $s_x\in\{s_1\}\us{c}$ conditioned on an outcome $i$ of suitable local measurements on the remaining $n-2$ subsystems is sufficient for the certification of GME of the system.
The number of bipartitions needed to be checked is thus $n-1$.
\end{theorem}
\begin{proof}
To prove the theorem, let us first introduce a notation that simplifies and clarifies the procedure.
For any pair of subsystems $s_x$ and $s_y$, on which we want to localize entanglement,  
we can denote by $(S\s{L},[s_x|s_y],S\s{R})$ a bipartition with respect to which the entanglement is being certified.
Here, $S\s{L}$ and $S\s{R}$ are the collection of all subsystems except $s_x$ and $s_y$ that belong to the left and right parties, respectively, and on which the conditioning takes place.
We also call the block $[s_x|s_y]$ the conditional bipartition.

By using Lemma~\ref{lem:condEnt}, for a fixed choice of subsystems $s_x$ and $s_y$, the localizable entanglement of any such bipartition is equivalent to that of the unique bipartition $([s_x|s_y],S'\s{R})$ where $S'\s{R}=S\s{R}\cup S\s{L}$. This follows directly by assuming that Alice is holding the subsystem $A=s_x$ while Bob holds $B=s_y\cup S\s{L}\cup S\s{R}$ and is able to make local measurements
on $C=S'\s{R}$.

Thus, one can think of the fixed block $[s_x|s_y]$ sliding through to generate bipartitions of the form $(S\s{L},[s_x|s_y],S\s{R})$ without changing the entanglement
properties of the resulting bipartition.
We symbolically denote this equivalence of bipartitions in terms of their localizable entanglement 
as $([s_x|s_y],S'\s{R}) \cong (S\s{L},[s_x|s_y],S\s{R})$.

Now, consider the bipartition $(S\s{L},[s_x|s_y],S\s{R})$.
Since we are interested only in a sufficient condition, we can safely do the replacement $s_x\mapsto s_1$ obtaining the bipartition $(S'\s{L},[s_1|s_y],S\s{R})$ with $S'\s{L}=\{s_1,\dots,s_j\}\setminus \{s_1\}$.

It is thus immediate that $(S\s{L},[s_x|s_y],S\s{R}) \cong (S'\s{L},[s_1|s_y],S\s{R})$. Next, recalling that the measurements are local, we can combine $S'\s{L}$ with $S\s{R}$ and write $(S'\s{L},[s_1|s_y],S\s{R}) \cong ([s_1|s_y],S'\s{L}\cup S\s{R})$, hence $(S\s{L},[s_x|s_y],S\s{R}) \cong ([s_1|s_y],S'\s{L} \cup S\s{R})$.

It follows from the above argument that bipartitions $\left\{([s_1|s_y],S\s{R}):\forall s_y\in\{s_2,\dots,s_n\}\right\}$ cover the set of all possible bipartitions, in the sense that detection of their conditional entanglement is sufficient to certify entanglement in all possible bipartitions, hence GME. 
Since this set has $n-1$ elements corresponding to the $n-1$ pairings of $s_1$ with other subsystems, the detection of GME only requires a linear number of bipartitions. 
\end{proof}
It is worth mentioning that, despite the systematic reduction in the number of bipartitions needed for the certification of GME, in general, there exist GME states the entanglement of which cannot be detected through localizing procedure as described in Theorem~\ref{th:measNo}.
In particular, it was shown by Mi\v{c}uda \textit{et al\ARB{.}}~\cite{Micuda2017} that there exist three-qubit mixed states that are genuinely multipartite entangled and, yet, no measurement on any of the subsystems leads to a bipartite entangled state of the remaining subsystems.
An example state with this property is the Werner state~\cite{Werner1989} 
\begin{equation}
    \varrho\s{W}={1-p \over 8} I + p\ket{\mathrm{GHZ}_+\oprod{3}}\bra{\mathrm{GHZ}_+\oprod{3}}
\end{equation}
for ${1 \over 5} < p \leq {1 \over 3}$.
Moreover, very recently, a related result within the context of device-independent GME certification was obtained independently by Zwerger \textit{et al.}~\cite{Zwerger2019}.
There, it was shown that all pure GME states contain conditional bipartite entanglement.
As a corollary, one can say that if a pure state is GME, then its GME can necessarily be certified using our conditional GME witnessing technique.
Interestingly, it is possible to reformulate Theorem~\ref{th:measNo} in terms of neighbouring subsystems of an $n$-partite system.
\begin{theorem}\label{th:measNoNeighbour}
Consider an $n$-partite quantum system of any dimension.
Then, the certification of entanglement in all conditional bipartitions $[s_x|s_{x+1 (\mod{}~n)}]$, with $1\leq x \leq n$, conditioned on an outcome $i$ of suitable local measurements on the remaining $n-2$ subsystems is sufficient for the certification of GME of the system.
The number of bipartitions needed to be checked is thus $n-1$.
\end{theorem}
\begin{figure}[h!]
  \includegraphics[width=0.5\columnwidth]{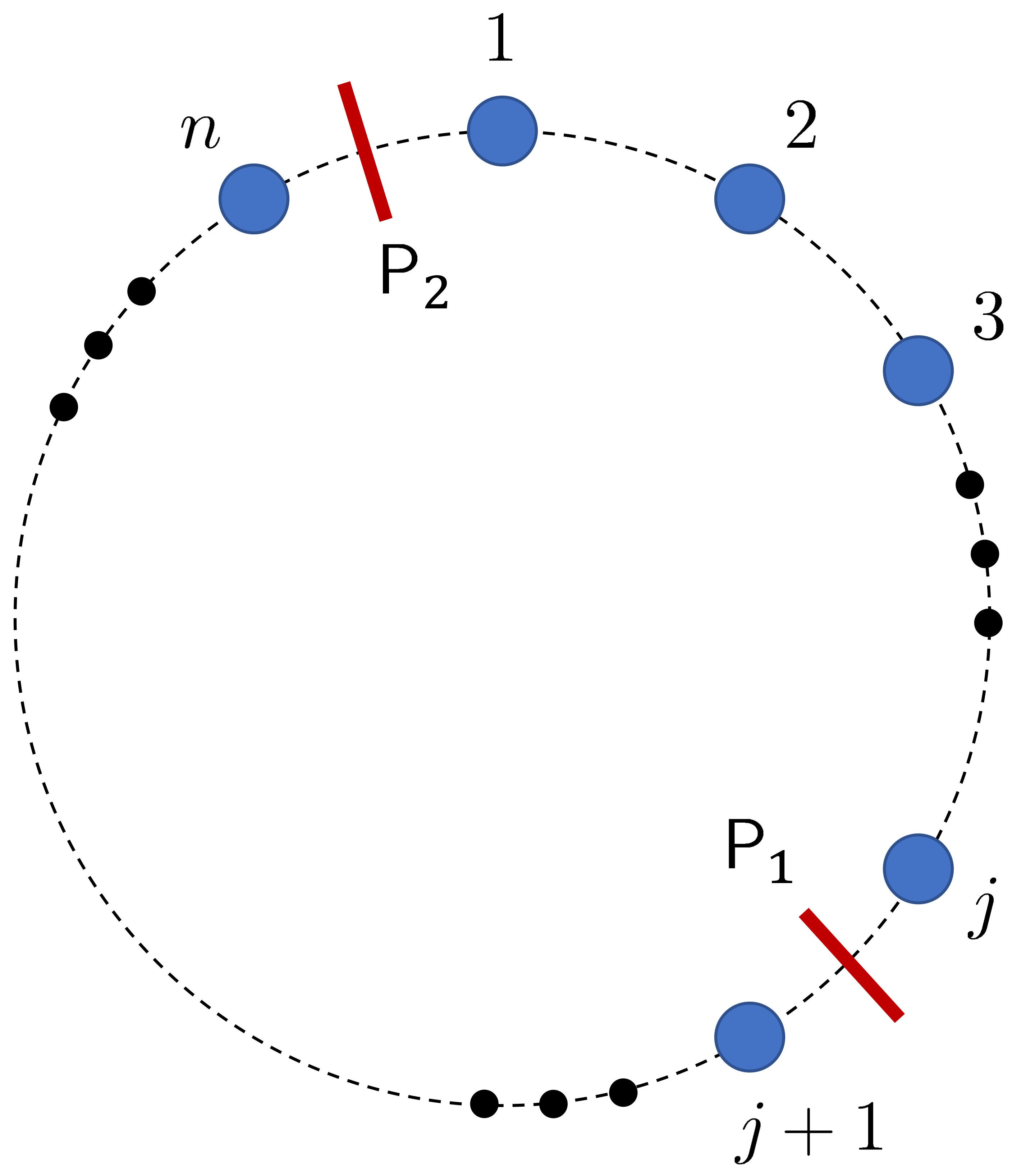}
  \caption{Schematic of the bipartition $(s_1,\dots,s_j|s_{j+1},\dots,s_n)$ in neighboring subsystems for $n$ qubits arranged on the circumference of a circle.
  The placements of separators $\pp_1$ and $\pp_2$ suffice to specify any given bipartition of the system.}
  \label{fig:proofT2}
\end{figure}
\begin{proof}
It is possible to give an intuitive pictorial proof of this theorem.
Suppose that, as shown in Fig.~\ref{fig:proofT2}, the $n$ subsystems are arranged on a circle so that subsystems $s_1$ and $s_n$ are neighbours.
A given bipartition $(s_1,\dots,s_j|s_{j+1}\dots,s_n)$ can be identified with the placements of two separators $\pp_1$ and $\pp_2$. 
Now, suppose that $\pp_1$ is placed between two subsystems $s_j$ and $s_{j+1}$ and suppose that, conditioned on local measurements on the remaining subsystems, we find that  $s_j$ and $s_{j+1}$ are entangled.
It is immediate, using Lemma~\ref{lem:condEnt}, that due to the locality of measurements the placement of the separator $\pp_2$ is irrelevant, so that a nonzero LE between $s_j$ and $s_{j+1}$ implies entanglement within all bipartitions that can be obtained by displacing separator $\pp_2$.

Finally, since the placements of $\pp_1$ and $\pp_2$ are sufficient to specify all bipartitions, we conclude that the only relevant degree of freedom is the placement of $\pp_1$.
Finally, the latter can be done in $n-1$ ways, hence the result.
\end{proof}

Lastly, let us note that in Theorems~\ref{th:measNo} and~\ref{th:measNoNeighbour} we have focused on local measurements.
However, it is possible to derive variants of them in which joint entangled measurements on specific subsystems are allowed.
Due to the fact that joint measurements generally increase the localization power, it is expected to obtain more powerful but rather complicated entanglement criteria from such considerations.

\subsection{Conditional entanglement witnessing} \label{subsec:condentwitness}

Given an $n$-partite system, there are many ways to verify the entanglement localized between two subsystems by measuring
the rest, and then using either Theorem~\ref{th:measNo} or~\ref{th:measNoNeighbour} to show GME.
Additionally, we require the detection of each nonzero LE to be efficient for each bipartition. 
A particular approach to do the latter is witnessing~\cite{Amaro2018,Amaro2020}.
The collection of all quantum states with unlocalizable entanglement with respect to a specific conditional bipartition $[s_x|s_y]$ form a closed convex set denoted by $\sep_{[s_x|s_y]}$. 
It is worth emphasizing that the states in $\sep_{[s_x|s_y]}$ are not necessarily biseparable states, rather they are states that reduce to separable bipartite states upon any local measurements on $\{s_x,s_y\}\us{c}$.
It thus follows that $\sep_{[s_x|s_y]}\supset \sep_{\ind_1|\ind_2}$ for any $\ind_1\ni s_x$ and any $\ind_2\ni s_y$.
For every $n$-partite state $\varrho\notin\sep_{[s_x|s_y]}$, that is a state with LE with respect to the conditional bipartition $[s_x|s_y]$, there exists a bounded Hermitian operator $W_{[s_x|s_y]}$ of the form
\begin{equation}\label{eq:condwitness}
    W_{[s_x|s_y]}=W_{s_x|s_y}\bigotimes_{z\in \{s_x,s_y\}\us{c}} M_{s_z:i},
\end{equation}
with $M_{s_z:i}$ being the effect corresponding to the outcome $i$ of some specific
measurement on the subsystem $s_z$, such that
\begin{equation}
\begin{split}
    \forall \sigma\in\sep_{[s_x|s_y]}\quad &\Tr \left(W_{[s_x|s_y]}\sigma\right) \geq 0,\\
    & \Tr \left(W_{[s_x|s_y]}\varrho\right) <0.
\end{split}
\end{equation}
We call the operator $W_{[s_x|s_y]}$ a \textit{conditional entanglement witness}, the existence of which is guaranteed by (i) the assumption that $\varrho$ contains LE between $s_x$ and $s_y$, namely $\varrho\notin\sep_{[s_x|s_y]}$, and (ii) as described in Sec.~\ref{subsec:witness}, by the Hahn-Banach separation theorem, there exists a witness $W_{s_x|s_y}$ for the entanglement concentrated between these two subsystems.
Furthermore, for each pair $s_x$ and $s_y$, the bipartite witness $W_{s_x|s_y}$ can be constructed from a test operator $L_{s_x|s_y}$ using the approach delineated in Sec.~\ref{subsec:witness}. 
It is worth emphasizing that, determining the separability bound for the test operator $L_{s_x|s_y}$ to be used in the construction of $W_{s_x|s_y}$ in Eq.~\eqref{eq:condwitness} requires an optimization over biseparable states of systems $s_x$ and $s_y$.
Hence, the complexity of determining this bound is the same as generic entanglement witnessing, i.e.~it is efficiently decidable~\cite{Eisert2004,Doherty2005}.
Combining conditional entanglement witnessing with either of Theorems~\ref{th:measNo} and~\ref{th:measNoNeighbour} thus gives an efficient technique for verification of GME that maintains the practicality of the entire process with additional robustness.
We call this approach \textit{conditional GME witnessing}.

In the rest
of this paper, we apply the conditional GME witnessing to QEC stabilizer measurement circuits to characterize their performance in terms of their power for the creation of 
GME states in the presence of errors, different sources of noise and inefficiencies.
We show that with a linear (in the number of qubits) increase in the complexity of the approach, we obtain significant noise robustness in our technique compared to the standard single witness entanglement tests~\cite{Toth2005}. 
Our claim can be pictorially be understood as shown in Fig.~\ref{fig:witnesscompare}.
Any standard GME witness can only detect GME states that do not belong to the closed convex hull of sets of separable states with respect to each bipartition.
There, however, exist GME states that belong to this convex hull.
Using conditional GME witnessing, in general, we use a linear number of witnesses to detect such GME states and obtain a higher resolution in the detection of GME.

\section{Noisy trapped-ion circuits} \label{sec:noisy_trappedions}

Starting with the pioneering proposal of Cirac and Zoller for universal quantum computation 
using the internal
states of trapped ions as qubits, and the vibrational modes as a quantum bus to mediate entangling gate operations~\cite{PhysRevLett.74.4091, SchmidtKaler2003}, systems of trapped atomic ions in radio-frequency potentials are nowadays considered to be among the most promising quantum information processors~\cite{Ladd2010, doi:10.1063/1.5088164}. 
Over the years, various quantum protocols have been realized in different trapped-ion platforms~\cite{HAFFNER2008155, doi:10.1063/1.5088164} including 
small-scale QEC algorithms~\cite{Chiaverini2004,Schindler1059},
a topologically encoded qubit based on the 7-qubit color code~\cite{Nigg302}, fault-tolerant error detection~\cite{Linkee1701074}, deterministic correction of qubit loss~\cite{Stricker2020}, and the first entangling gate at the level of logical qubits~\cite{erhard2020entangling}.
One of the current quests in trapped-ion quantum computing is to scale up these prototype processors towards larger-scale systems capable of taking full advantage of QEC routines of FT~\cite{PhysRevX.7.041061,PhysRevA.99.022330,PhysRevA.100.062307, murali2020architecting,PhysRevA.100.032325,Trout_2018}.

\subsection{Compilation into native ion-trap gates and operations}\label{sec:into_nativegates}

Whilst a large-scale FT implementation of QEC for universal quantum computations is beyond the reach of near-term devices, small-scale FT-QEC protocols can already  be run on some of the current trapped-ion platforms. 
Even though they still consist of a reduced number of qubits where some decoherence is unavoidable, they deliver long coherence times~\cite{Wang2017,PhysRevLett.113.220501} and high-fidelity single and two-qubit gates~\cite{PhysRevLett.117.060504,PhysRevLett.119.150503,doi:10.1116/1.5126186, Ballance2015} 
that are essential to test the performance of small distance FT-QEC codes or its building blocks.
We thus focus on this platform to demonstrate the robustness and efficiency of our entanglement characterization method.

The  native trapped-ion entangling gates in current designs are not CNOT gates~\cite{nielsen00}. 
Rather, they are based on state-dependent dipole forces and effective spin-spin interactions~\cite{PhysRevLett.82.1835,PhysRevA.62.022311,Sackett2000,Leibfried2003}. 
In this work, we shall be concerned with the so-called $Z$-state-dependent forces, which give rise to the entangling $ZZ$-gates~\cite{Leibfried2003,PhysRevLett.119.150503,doi:10.1116/1.5126186}, namely
\begin{equation}
\label{eq:zz}
    U^{\rm ZZ}_{\rm{ij}}(\theta)=e^{-\ii\frac{\theta}{2}Z_{i}Z_{j}},
\end{equation}
in which $ Z_i$ and $Z_j$ are the Pauli matrices of the corresponding $i$th and $j$th qubits involved in the gate, respectively, and $\theta$ is the corresponding pulse area.
In addition to $ZZ$-gates, which are fully entangling for $\theta=\pi/2$, i.e.~ they map product states onto GHZ-type entangled states, we shall also consider single-qubit rotations. In particular, rotations of the form
\begin{equation}
    R^{Z}_{i}(\theta)=e^{-\ii\frac{\theta}{2}Z_i}
\end{equation}
are obtained by local ac-Stark shifts ~\cite{Poschinger_2009}.
The parallel rotations of ion states within the equatorial plane of the Bloch sphere are obtained by a simultaneous
driving of the carrier transition of the ions in the laser focus,
\begin{equation}\label{eq:XYrotation}
    R^{\perp}_{\phi}(\theta)=e^{-\ii\frac{\theta}{2}\sum_i(\cos\phi X_i+\sin\phi Y_i)},
\end{equation}
where $R^{\perp}_{0}(\theta)$ and $ R^{\perp}_{\pi/2}(\theta)$ correspond to rotations around $x$ and $y$ axes of the Bloch sphere, respectively.
For instance, by setting $\phi=0$ and $\theta=\pi$, Eq.~\eqref{eq:XYrotation} represents single-qubit $\pi$ pulses applied to all illuminated ions. 
We note that, although these rotations act on all illuminated ions, the equatorial rotations can be applied to a particular set of ions by using spin-echo-type refocusing pulses that interleave these rotations with the addressable $Z$-type rotations~\cite{Nebendahl2009}.

Although this collection of gates is not the standard one in quantum computation~\cite{nielsen00}, it is a universal gate set so that any quantum algorithm can be decomposed into a sequence of these elementary operations~\cite{Nebendahl2009,Schindler2013}.
In Fig.~\ref{fig:native_gate_scheme}, we show the correspondence between the standard universal gate set and certain sequences of native trapped-ion gates.
In Fig.~\ref{fig:4+1SX_ZZ_circ}, on the other hand, we present the compiled circuit of Fig.~\ref{fig:stabiliser_readout}.~a into the native trapped-ion gate set.
\begin{figure}[t!]
  \includegraphics[width=1\columnwidth]{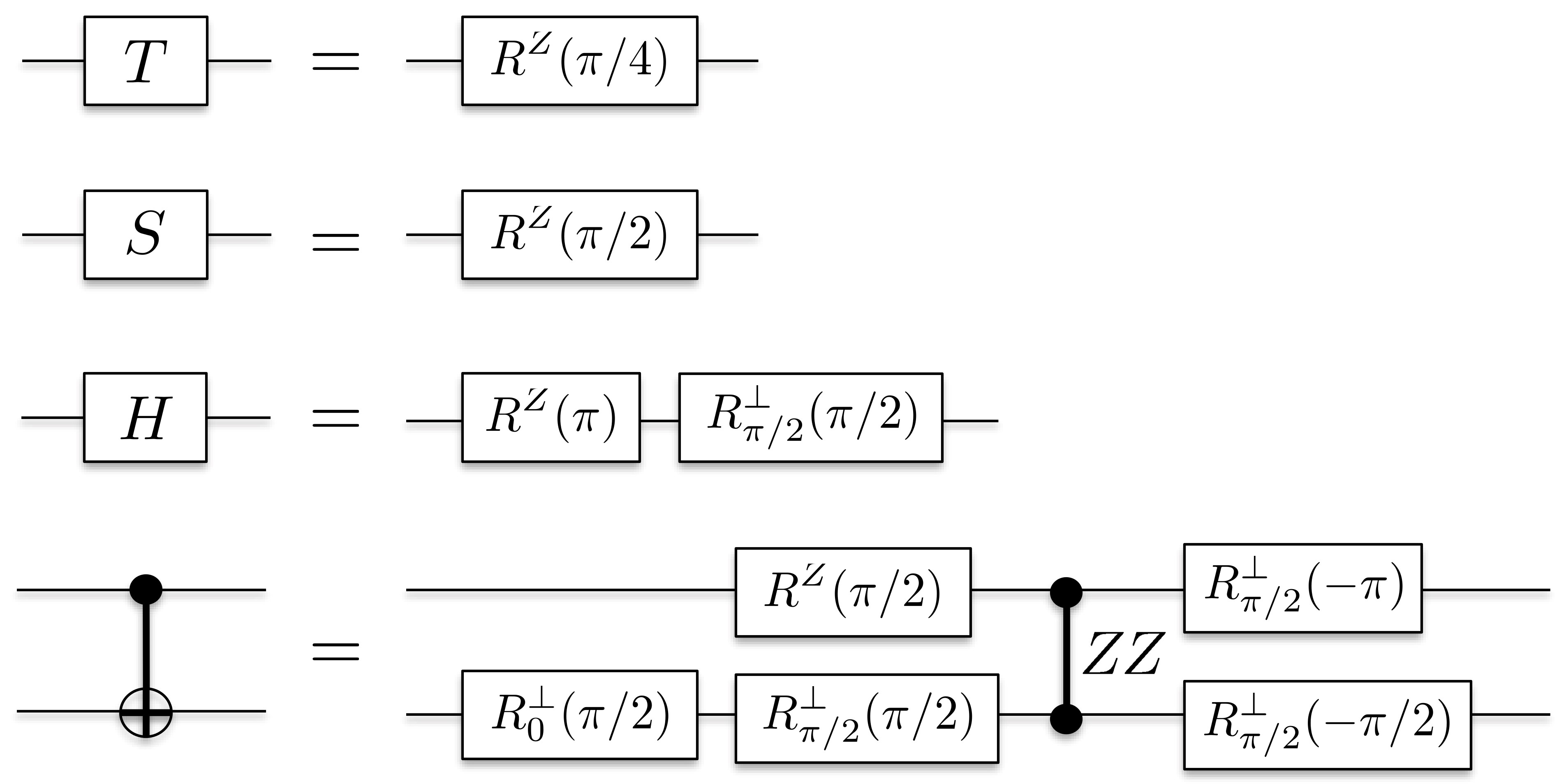}
  \caption{Correspondence between the standard and trapped-ion universal gate sets enabling quantum computation.
  The solid line that joins two filled circles represents the $ZZ$-gate $U^{\rm ZZ}_{ij}(\pi/2)$ in Eq.~\eqref{eq:zz}.}
  \label{fig:native_gate_scheme}
\end{figure}
\begin{figure}[h]
  \includegraphics[width=1\columnwidth]{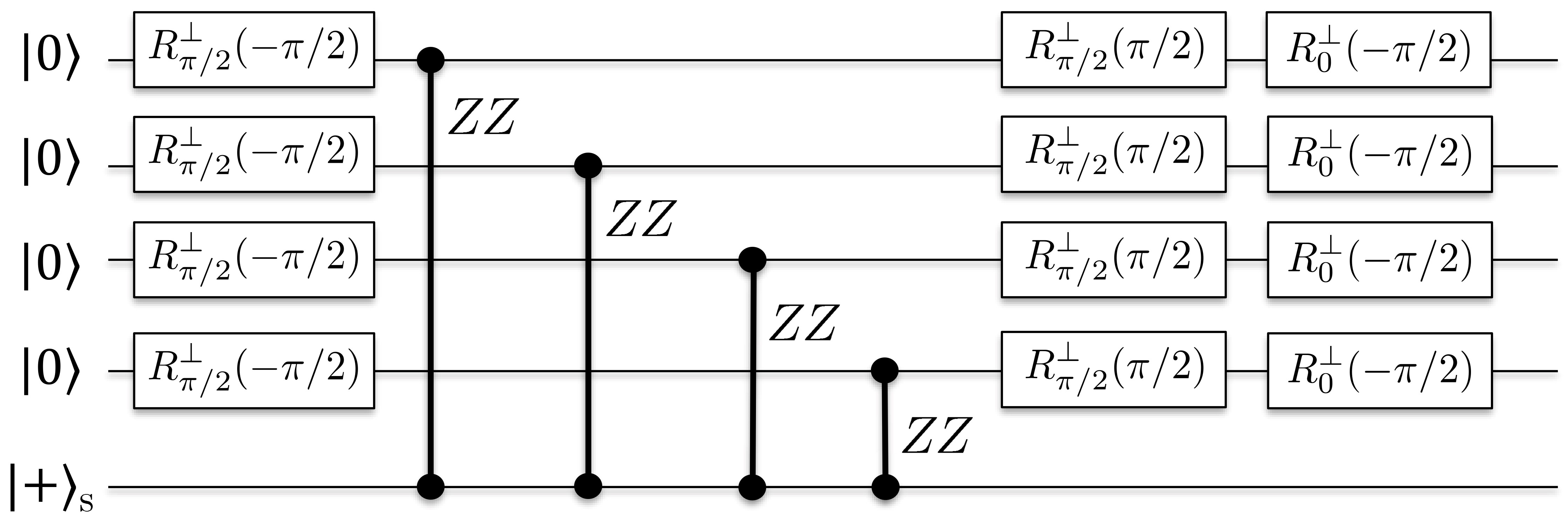}
  \caption{Non-FT circuit using the native trapped-ion gates.}
  \label{fig:4+1SX_ZZ_circ}
  \vspace{0.5cm}
  \includegraphics[width=1\columnwidth]{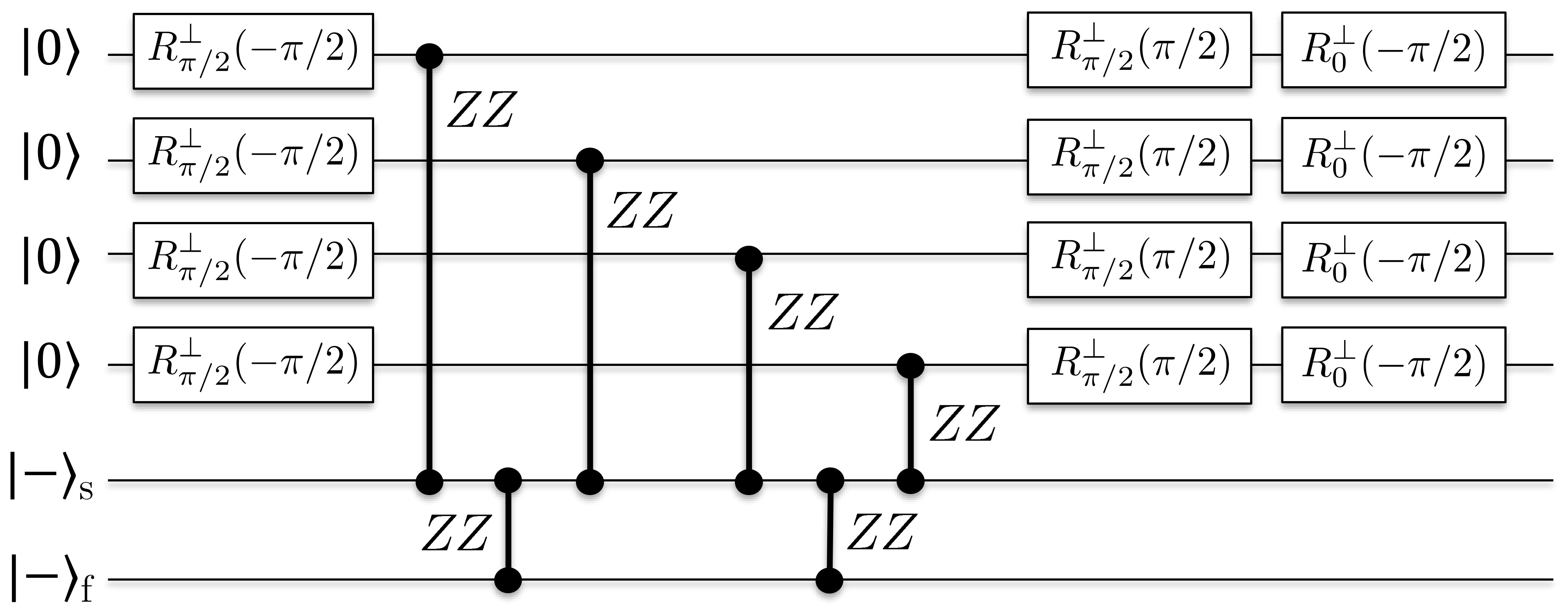}
 \caption{Flag-based FT circuit using the native trapped-ion gates.}
  \label{fig:4+2SX_ZZ_circ}
\end{figure}
The entangling $ZZ$-gates for $\theta=\pi/2$ between the $i$th and $j$th qubit can be written as $U^{\rm ZZ}_{i,j}(\pi/2)=(\mathbb{I}-\ii Z_{i}Z_{j})/\sqrt{2}$.
Similarly one can write the single-qubit rotations for qubit $i$th as $R^{\perp}_{\pi/2,i}(\pm\pi/2) = (\mathbb{I}\mp \ii Y_{{ i}})/\sqrt{2}$ and  $R^{\perp}_{0,i}(-\pi/2 ) = (\mathbb{I}+ \ii X_{i})/\sqrt{2}$.
After some algebra, the output state of the circuit in Fig.~\ref{fig:4+1SX_ZZ_circ} is obtained as
\begin{equation}\label{eq:output_5qubit_ZZ}
    \ket{\psi\s{out}}={\ket{0}\oprod{5} + \ket{1}\oprod{5} \over \sqrt{2}}=\ket{\mathrm{GHZ}\oprod{5}_{+}},
\end{equation}
which, as expected, is identical to the 5-qubit GHZ state in Eq.~\eqref{eq:output_5qubit} generated with the circuit comprising CNOT gates.
Similarly to the non-FT scenario with CNOT gates, we may analyze the propagation of a dangerous single $X$-error from the syndrome qubit to the data qubits through the $ZZ$-gates, the details of which are given in Appendix~\ref{app:ZZerror_propagate}.
It follows that, in this case too, the syndrome readout will not reveal any information about the propagation of the error. 
This leads us to the flag-based circuit shown in Fig.~\ref{fig:4+2SX_ZZ_circ}, which is the compiled version of Fig.~\ref{fig:stabiliser_readout}.~b into the native trapped-ion gate set. It is straightforward to calculate the output of this circuit as
\begin{equation}
    \ket{\psi\s{out}}={(\ket{0}\oprod{5} + \ket{1}\oprod{5})\ket{+}\s{f} \over \sqrt{2}}=\ket{\mathrm{GHZ}\oprod{5}_{+}}\ket{+}\s{f}.
    \label{eq:output_stabilizers_ZZ}
\end{equation}
One realizes that the sequential application of two $ZZ$-gates between the syndrome qubit $\rm s$ and the flag qubit $\rm f$, $U^{\rm ZZ}\s{sf}(\pi/2)U^{\rm ZZ}\s{sf}(\pi/2)\ket{-}\s{s}\ket{-}\s{f}$, leaves the two qubits in the product state $\ket{+}\s{s}\ket{+}\s{f}$. 
A dangerous Pauli $X$-error in the syndrome that produces a cascade of errors in the data, however, will flip the flag back into $\ket{-}\s{f}$.
Hence, a readout of the qubit~$\rm f$ ``flags'' the propagation of multiple errors that would not be possible in a non-FT scheme.
\begin{figure}[h]
  \includegraphics[width=1\columnwidth]{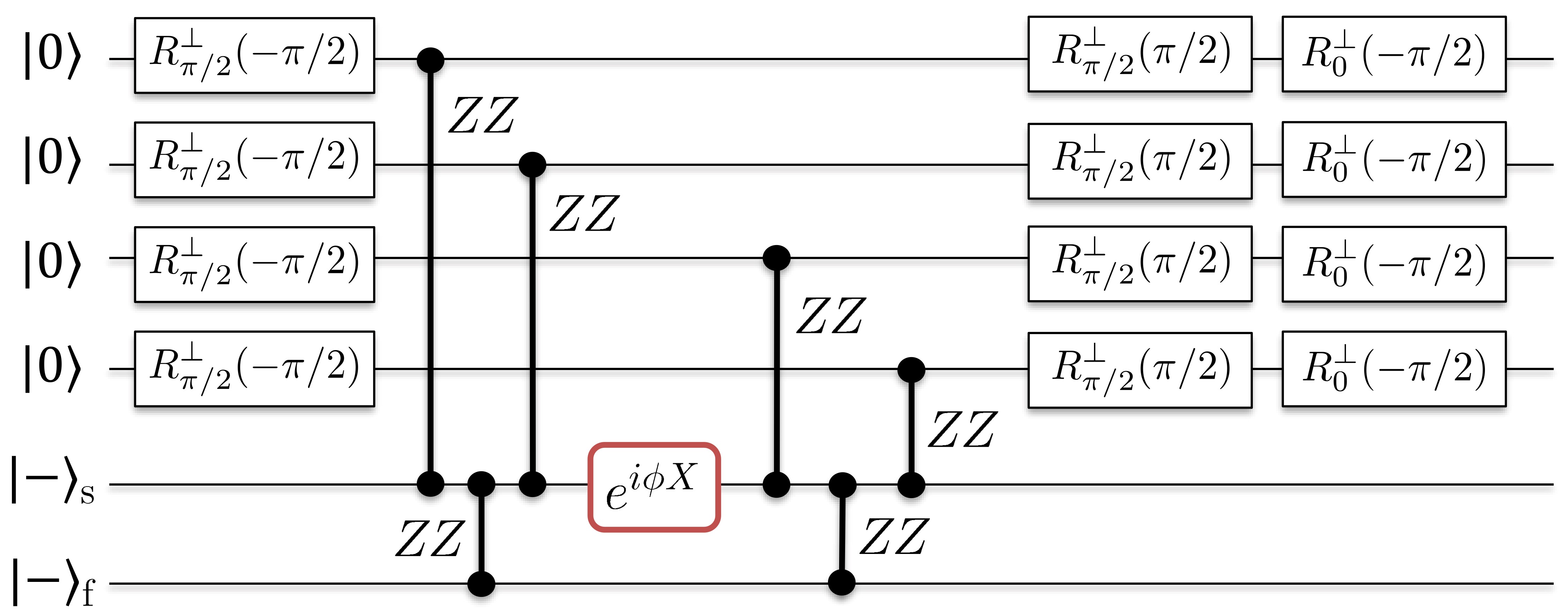}
  \caption{Flag-based FT stabilizer measurement circuit using the native trapped-ion gates. Here, a coherent-rotation error $e^{i\phi X}$ is injected in a specific position to create a 6-qubit entangled state in the flag-based FT stabilizer measurement.} 
  \label{fig:4+2SX_ZZ_circwit}
\end{figure}
In parallel with Fig.~\ref{fig:ideal_circuits}.~d, whenever the syndrome qubit suffers from a coherent-rotation error $\exp{(i\phi X)}$ between the second and third syndrome-data $ZZ$-gates, as shown in Fig.~\ref{fig:4+2SX_ZZ_circwit}, the ideal output state can be easily computed as
\begin{equation}\label{eq:Native_FT_6qubit_CohRot}
\begin{split}
    \ket{\psi\s{out}} = & \cos{\phi}\ket{+}\s{f}\ket{\mathrm{GHZ}\oprod{5}_+} +  \ii\sin{\phi} \ket{-}\s{f} Y\s{s}X_{3}X_{4} \ket{\mathrm{GHZ}\oprod{5}_+}.
\end{split}
\end{equation}
A projective measurement of the flag again unveils a cascaded error in the syndrome and data qubits. A similar procedure can be implemented using the $g_z$ generator as presented in Appendix~\ref{app:gz_ideal_ZZ}.
 
\subsection{Noise models}\label{subsec:noisemodels}

Recall from Sec.~\ref{subsec:witness} that, in general, the output of an ideal plaquette is a GME state.
Hence, to assess the robustness of the conditional GME witnessing method in a trapped-ion platform, we implement different noise models into the circuit, and compare the performance of conditional GME witnessing against the standard single witness (i.e. fidelity measurement) approach in detecting the potential GME output state.

In this regard, we use two different circuit-noise models and one measurement error scheme as leading errors in the stabilizer circuits for each of non-FT and flag-based FT scenarios.
As the first circuit noise model we use a simplified and phenomenological one in which the preparation of the input state and all the gates are assumed to be perfect while the decoherence
of qubits is simulated by an independent single-qubit depolarizing channel~\cite{nielsen00} applied to each qubit right before their measurements.
Our second noise model is motivated by the realistic experimental details discussed in Ref.~\cite{doi:10.1116/1.5126186} wherein shuttling techniques are employed to carry out the two-qubit gates.
Thus, we consider the case wherein the preparation of the input state is perfect while the two-qubit gates in the circuit suffer from independent depolarizing noise.

The measurement errors, modelled as independent classical bit-flip errors in the measured bits, are present in both scenarios and for all qubits.
We assume that dephasing effects in idle qubits are depreciated compared to the two-qubit errors in entangling gates.
This is a realistic assumption given the recent experiments on the use of stable magnetic fields which lead to a negligible dephasing of the qubits during the idle-time intervals~\cite{Ruster2016}. 
Moreover, in the considered experimental layout~\cite{doi:10.1116/1.5126186}, single-qubit gates also have a negligible error in comparison to the two-qubit and measurement errors, and can be thus neglected.
We do assume here that the error rates per two-qubit gate increase with the depth of the circuit as a consequence of the shuttling-based approach used in quantum charge-coupled devices~\cite{Kielpinski2002,PhysRevLett.109.080501,PhysRevA.90.033410}.

\paragraph{\it{Independent depolarizing noise.}---}\label{para:pheno_noise}

This error channel is used only for the phenomenological model. 
It consists of independent depolarizing channels of the form
\begin{equation}\label{eq:noisedepo1}
\begin{split}
   \epsilon\us{dip}_{i}(\varrho)=(1-p)\varrho+ \frac{p}{3}\sum_{\sigma\in\{X,Y,Z\}}\sigma_{i}\varrho\sigma_{i},
\end{split}
\end{equation}
acting on each of the qubits at the end of the circuit sequence and just before the measurements.
Here, $p$ is probability for a depolarizing error to occur on the $i$th qubit and we consider it to be the same for all qubits.

\paragraph{\it{Two-qubit depolarizing noise.}---}\label{para:circuit_noise}
In a trapped-ion platform, ions can be shuttled in and out of the laser interaction zone in which the quantum operations and readouts are carried out~\cite{doi:10.1116/1.5126186}.
This process may excite the ions' vibrational modes that affect the fidelity of the two-qubit gates depending on their order in the shuttling sequence or, equivalently, the time step at which the gate is applied.
This detrimental effect can be modelled in a conservative manner by a two-qubit depolarizing noise after the application of each two-qubit entangling gate $U^{\rm ZZ}_{ij}$.

For each pair of ions involved in a $ZZ$-gate, they may undergo 15 possible one and two-qubit Pauli errors so that the total error channel is described as
\begin{equation}\label{eq:noisedepo2}
\begin{split}
    \epsilon\us{dip}_{i,j}(\varrho)=[1-p(r,t)]\varrho+ \frac{p(r,t)}{15}\sum_{\sigma,\kappa\in\{I,X,Y,Z\}}\sigma_{i}\kappa_{j}\varrho\kappa_{j}\sigma_{i}
\end{split}
\end{equation}
where $i$ and $j$ with $i\neq j$ denote the active ions involved in each two-qubit entangling gate, $p$ is the  error probability, and the sum on the right runs over the 6 nontrivial single-qubit and the 9 nontrivial two-qubit Pauli operators.
We assume a time (or circuit depth) dependent exponential growth in the error probability defined as $p(r,t)=p(1+r)^{t}$,
where $r$ and $t$ are the error-growth rate per gate and the time step at which $ZZ$-entangling gates are applied, respectively.
It is worth to clarify that $t$ is not the actual elapsed physical time in an experimental implementation rather it is an index taking values $t\in\mathbb{N}$ specifying the order of entangling gates in the circuit.
Hence, we assume the time increment between every two entangling gate to be $\Delta t = 1$.
Upon using this model, we also consider three error-growth rates of $r=0$, $r=0.1$, and $r=0.2$, noting that an $r=0.2$ implies that the $4$th and $6$th two-qubit gates are about $2$ and $3$ times worse than the first one, respectively.

\paragraph{\it{Bit-flip measurement errors.}---}\label{para:mes_noise}

Entanglement witnessing procedures with and without conditioning involve Pauli measurements of both data and ancillary qubits, which may also suffer from imperfections.
We use a simple model for measurement errors, namely independent bit-flips.
In a measurement of a single-qubit Pauli operator $\sigma\in\{X,Y,Z\}$ the corresponding error-free positive operator-valued measure (POVM)~\cite{nielsen00} is specified by the set of operators
$\{E^\sigma_+={I + \sigma \over 2},E^\sigma_-={I - \sigma\over 2}\}$.
For each measurement outcome, a bit-flip error taking place with probability $p\s{me}$ thus gives rise to the POVM effects
\begin{equation}
\begin{split}
    e^{\sigma}_+ & = (1-p\s{me}) E^\sigma_+ + p\s{me} \kappa^\sigma E^\sigma_+ \kappa^\sigma \\
    & = (1 - p\s{me}) E^\sigma_+ + p\s{me} E^\sigma_-,
\end{split}
\end{equation}
and
\begin{equation}
\begin{split}
    e^{\sigma}_- & = (1 - p\s{me}) E^\sigma_- + p\s{me} \kappa^\sigma E^\sigma_- \kappa^\sigma\\
    & = (1 - p\s{me}) E^\sigma_- + p\s{me} E^\sigma_+. 
\end{split}
\end{equation}
Here, $\kappa^\sigma\in\{X,Y,Z\}$ is the (measurement dependent) error operator such that $\kappa^\sigma E^\sigma_+ \kappa^\sigma = E^\sigma_-$ and $\kappa^\sigma E^\sigma_- \kappa^\sigma=E^\sigma_+$.
It is also obvious that the effect of measurement errors can instead be described by a quantum channel acting on the individual qubits such that
\begin{equation}\label{eq:noisebitflip}
\begin{split}
    \epsilon^\sigma_{i}(\varrho)=(1-p\s{me})\varrho+p\s{me}\kappa^\sigma\varrho\kappa^\sigma.
\end{split}
\end{equation}

\section{Robustness of conditional GME witnessing}\label{sec:robustnes_GME_trappedions}

We are now ready to examine the robustness and efficiency of our conditional GME witnessing technique compared to the standard GME certification approaches used within the literature.
To this end, we consider noisy non-FT and flag-based FT plaquette measurement circuits and witness the resulting entangled state using imperfect measurements in three different ways:
(i) \textit{standard witnessing} with a GHZ projector (i.e. a single  witness) requiring an exponential number of measurements, that is, by measuring fidelity to the ideal output state,
(ii) \textit{standard--linear witnessing} (SL) with a single witness that relies on only a linear number of measurements using the proposed witnesses of T\'{o}th and G\"{u}hne~\cite{Toth2005}, and
(iii) our proposed \textit{conditional GME witnessing} requiring a linear number of witnesses and measurements.

\subsection{Conditional GME witnessing in the  noisy non-FT plaquette circuit}

Consider the 5-qubit non-FT $g_x$ circuit from Fig.~\ref{fig:4+1SX_ZZ_circ}. 
We implement each noise model of Sec.~\ref{subsec:noisemodels} on the circuit and witness the output entangled state using the three above-mentioned techniques.
Beginning with the standard witnessing, it was shown in Sec.~\ref{subsec:witness} (cf. Eq.~\eqref{eq:stdTestGHZ5}) that the $5$-qubit GHZ projector is a possible test operator for GME,
\begin{equation}\label{eq:stdTestGHZ5-rep}
L\s{GHZ\oprod{5}}=\ket{\mathrm{GHZ}\oprod{5}_{+}}\bra{\mathrm{GHZ}\oprod{5}_{+}}.
\end{equation}
It was also shown in Eq.~\eqref{eq:stdWitnessGHZ5} that the detection bound for this test operator is $1\over 2$ for any bipartition, i.e. given any state $\varrho$, an expectation value satisfying $\Tr \varrho (L\s{GHZ\oprod{5}}) > {1\over 2}$ implies the entanglement of $\varrho$ with respect to all bipartitions and thus its GME.
We have also shown in Eq.~\eqref{eq:GHZ5Stabilizers} that the 5-qubit GHZ projector can be written in terms of its $2^5$ stabilizers. 
Combining the two, we observe that the expectation value of $L\s{GHZ\oprod{5}}$ for any state is given by
\begin{equation}\label{eq:stdWitnessGHZ5_numerical}
   \expec{L\s{GHZ\oprod{5}}} = \Tr (\varrho L\s{GHZ\oprod{5}}) = {1 \over 2^{5}}\sum_{i=1}^{2^{5}}\expec{S_i}.
\end{equation}
It is also worth pointing out that $\expec{L\s{GHZ\oprod{5}}}$ equates to the fidelity of the state $\rho$ with the ideal output state of the 5-qubit non-FT plaquette circuit derived in Eq.~\eqref{eq:output_5qubit}.
For the standard witnessing of GME using the 5-qubit GHZ projector $L\s{GHZ\oprod{5}}$, it is thus sufficient to locally measure the stabilizers at the output of the plaquette circuit the number of which would in general scale exponentially with the size of the register.

The second approach we consider is the SL witnessing of GME using a witness constructed out of the $5$ generators of the $5$-qubit stabilizer subgroup given in Eq.~\eqref{eq:5GHZ_stabilizergroup}.
Two test operators introduced by T\'{o}th and G\"{u}hne~\cite{Toth2005} for this purpose are, after normalization, 
\begin{equation}\label{eq:TG1GHZ5}
    L\s{TG1} = \frac{1}{2}{I+g_5\over2}+{1\over 2}\prod_{i=1}^{4}{I+g_i\over 2}
\end{equation}
and
\begin{equation}\label{eq:TG2GHZ5}
    L\s{TG2} = {1\over 5}\sum_{i=1}^{5}g_i.
\end{equation}
The GME bounds for these witnesses are $l\s{TG1}={3\over 4}$ and $l\s{TG2}={4\over 5}$, respectively~\cite{Toth2005}.
Noting the form of the stabilizer generators, the obvious advantage of these witnesses is that they only require two measurement settings.

Finally, we implement our proposed conditional GME witness and show through this example that it is technically simple and, at the same time, powerful.
Using Theorem~\ref{th:measNo}, we focus merely on the four conditional bipartitions $[\syn|1]$, $[\syn|2]$, $[\syn|3]$, and $[\syn|4]$ and introduce the conditional test operators
\begin{equation}\label{eq:condWitness_5_numeric}
    L_{[\syn|x]}=\ket{\Bell}_{\syn|x}\bra{\Bell} \otimes \ketbra{+}\oprod{3}
\end{equation}
for $x\in\{1,2,3,4\}$.
Here, $\ket{\Bell}_{\syn|x}\bra{\Bell}$ is the projection of the qubit pair $\syn$ and $x$ onto the Bell state $\ket{\Bell}=(\ket{00}+\ket{11})/\sqrt{2}$
and $\ketbra{+}\oprod{3}$ denotes the projector of the remaining three qubits $\{\syn,x\}\us{c}$.
The bounds on the conditional test operators are easy to compute as we are effectively dealing with two-qubit systems.
For the case of interest here, it is known that
\begin{equation}
    l_{\syn|x}=\sup_{\sigma\in\sep_{\syn|x}} \Tr \left(\sigma \ket{\Bell}_{\syn|x}\bra{\Bell}\right) = {1 \over 2},
\end{equation}
where the supremum is taken over the set of all bipartite separable states of qubits $\syn$ and $x$.
Consequently, for any conditional bipartition $[\syn|x]$ with $x\in\{1,2,3,4\}$ and any given $5$-qubit state $\varrho$, an expectation value $\expec{L_{[\syn|x]}}>{1\over 2}$ for all $x\in\{1,2,3,4\}$ implies GME of the quantum state $\varrho$.
\begin{figure*}[ht!]
  \includegraphics[width=1.\textwidth]{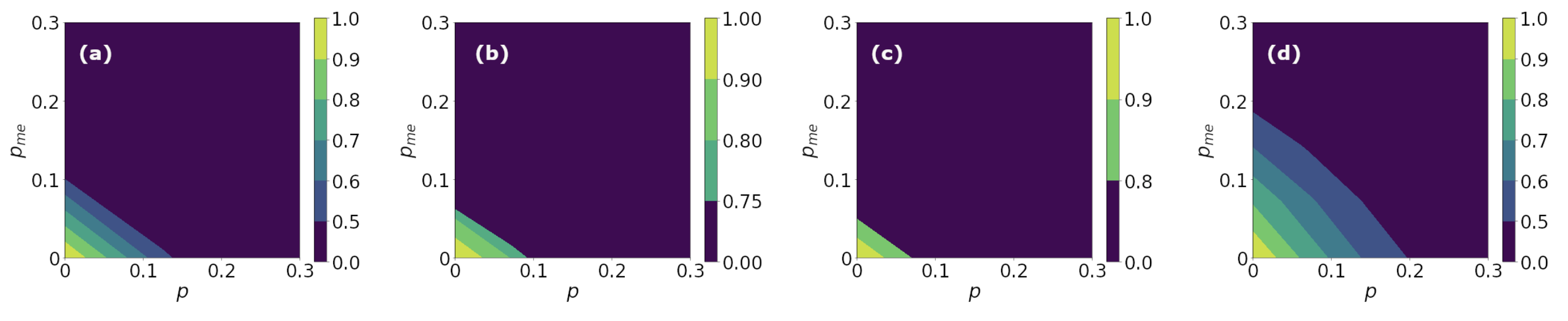}
  \caption{GME witnessing plots in a 5-qubit non-FT plaquette readout circuit under the phenomenological depolarizing noise model.
  Each contour plot represents, for different values of depolarizing-noise probability $p$ and measurement-error probability $p\s{me}$, the expectation values of (a) the standard GME witness $L\s{GHZ\oprod{5}}$ of Eq.~\eqref{eq:stdWitnessGHZ5_numerical}, (b) the first SL
  witness, $L\s{TG1}$ in Eq.~\eqref{eq:TG1GHZ5}, (c) the second SL
  witness $L\s{TG2}$ in Eq.~\eqref{eq:TG2GHZ5}, and (d) the worst-case conditional witnessing using $L_{[s|x]}$ in Eq.~\eqref{eq:condWitness_5_numeric} for $x=1$. 
  The colormap sidebars represent the witness bounds for each method.
  In all cases, the dark blue region identifies the noise values for which the GME tests are inconclusive.}
 \label{fig:contour_depoattheend}
\end{figure*}

Importantly for us, the Bell state is the one-dimensional code space of the stabilizer generators $\{XX,ZZ\}$.
Similar to Eq.~\eqref{eq:GHZ5Stabilizers} we thus have
\begin{equation}
    \ket{\Bell}_{\syn|x}\bra{\Bell}={I\s{s}I_x+X\s{s}X_x-Y\s{s}Y_x+Z\s{s}Z_x \over 4},
\end{equation}
implying that conditional GME witnessing can be performed solely by local Pauli measurements on the qubits.
Furthermore, for the implementation of conditional GME witnessing using $L_{[s|x]}$ of Eq.~\eqref{eq:condWitness_5_numeric} we need $3$ settings per bipartition, hence a total of $12$ measurement settings.
It is also noteworthy that, while we only consider the projection onto $\ketbra{+}\oprod{3}$, we can consider other outcome combinations of the $X$ measurements on these three qubits that give rise to either the same Bell state or the one orthogonal to it, namely $(\ket{00}_{\syn|x}-\ket{11}_{\syn|x})/\sqrt{2}$, that can be witnessed seamlessly with the same measurement settings.
In experimental implementations of the conditional GME witnessing, in general, more than one of the conditioning events can thus be used to keep the overall witnessing procedure efficient costing at most a linear overhead in the number of settings.

We now turn to evaluating the robustness of our technique versus other witnessing approaches in the literature.\\
\begin{figure}[h!]
  \includegraphics[width=0.8\columnwidth]{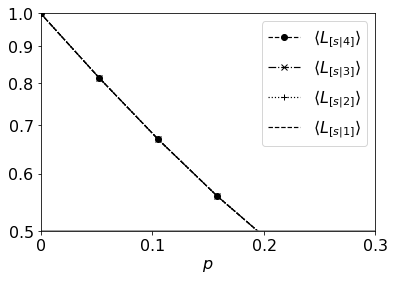}
  \includegraphics[width=0.8\columnwidth]{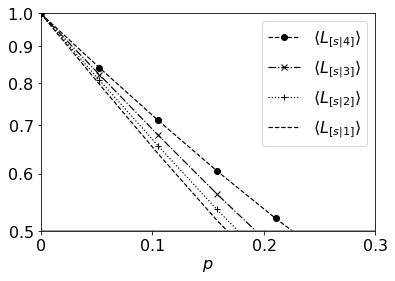}
  \caption{Expectation values of the conditional test operator of Eq.~\eqref{eq:condWitness_5_numeric} for four conditional bipartitions when the phenomenological depolarizing noise model (top panel) or the two-qubit depolarizing noise model (bottom panel) is applied to the non-FT readout of the plaquette.
  The $y$-axes are represented in log-scale whilst the $x$-axes represent the different error probabilities of the two depolarizing error models. 
  In both cases the measurement error probability $p\s{me}$ is assumed to be zero.
  In the bottom panel we have assumed an error-growth rate per gate of $r=0.2$.
  For the phenomenological depolarizing noise, due to the full symmetry of the error channel and the test operator with respect to different conditional bipartitions, the choice of the conditional bipartition is irrelevant.
  For the two-qubit depolarizing noise, on the other hand, the lowest performance belongs to the qubit pair that were involved in the first two-qubit gate, i.e. $\rm [s|1]$ (see the main text).}
  \label{fig:thresholds_5q}
\end{figure}
\subsubsection{Phenomenological depolarizing noise model}\label{sec:pheno_5q}
The contour plots of Fig.~\ref{fig:contour_depoattheend} present our results for the four witnessing techniques applied to the 5-qubit non-FT plaquette readout circuit (see Fig.~\ref{fig:4+1SX_ZZ_circ}) subjected to the noise channel of Eq.~\eqref{eq:noisedepo1} and the measurement noise of Eq.~\eqref{eq:noisebitflip}.
From left to right, the four subfigures represent the results obtained for (a) the standard witness of Eq.~\eqref{eq:stdTestGHZ5-rep},
(b) the SL witness of Eq.~\eqref{eq:TG1GHZ5}, (c) the SL witness of Eq.~\eqref{eq:TG2GHZ5}, and (d) our
conditional witness of Eq.~\eqref{eq:condWitness_5_numeric} for $x=1$.
The dark blue shaded areas in all plots represent states in which
their GME, if any, cannot be detected by the corresponding witness.
One can thus clearly appreciate the significant increase in the noise tolerance of conditional GME witnessing compared to the other witnessing methods that can be found in previous literature. 

In particular, for $p\s{me}=0$, the highest depolarizing error probability $p$ that can be afforded before losing the GME witnessing capabilities for the standard method (Fig.~\ref{fig:contour_depoattheend}.~a) is $p\approx 0.13$. 
For the first (Fig.~\ref{fig:contour_depoattheend}.~b) and the second (Fig.~\ref{fig:contour_depoattheend}.~c) SL witness
we find $p\approx 0.09$ and $p\approx 0.07$, respectively.
These values demonstrate the typical trade-off between the number of measurements and the robustness of witnesses.
For our conditional GME witness (Fig.~\ref{fig:contour_depoattheend}.~d), however, the error tolerance goes up to $p\approx 0.2$.
We thus observe the counterintuitive fact that a linear number of measurements does not necessary imply a loss of robustness in GME detection. 
Similarly, by fixing $p=0$ for the depolarizing noise, we notice that the measurement probability error $p\s{me}$ also reaches its highest threshold value for our conditional witnessing with $p\s{me} \approx 0.2$ compared to $p\s{me}\approx 0.1$ for the standard method.
We see the lowest robustness for the first and the second SL witness
with $p\s{me}\approx 0.07$ and $p\s{me}\approx 0.05$, respectively.

Recall that, in the conditional GME witnessing method for the 5-qubit circuit, according to Theorem~\ref{th:measNo}, it is needed to test entanglement in a total of four bipartitions.
In the top panel of Fig.~\ref{fig:thresholds_5q}, we present 
the behavior of the expectation value $\expec{L_{[\syn|x]}}$ with respect to the depolarizing error probability $p$ for different $x$s at the fixed measurement error probability $p\s{me}=0$.
It is evident that for this noise model the four bipartitions behave identically, as expected from the symmetries of the ideal state, the noise model exploited, and the conditional GME test used.
Therefore, in this case, all the bipartitions would exhibit the same threshold values as in Fig.~\ref{fig:contour_depoattheend}.~d. 

\subsubsection{Two-qubit depolarizing noise model}\label{sec:circ_5q}
\begin{figure*}[ht!]
  \includegraphics[width=1\textwidth]{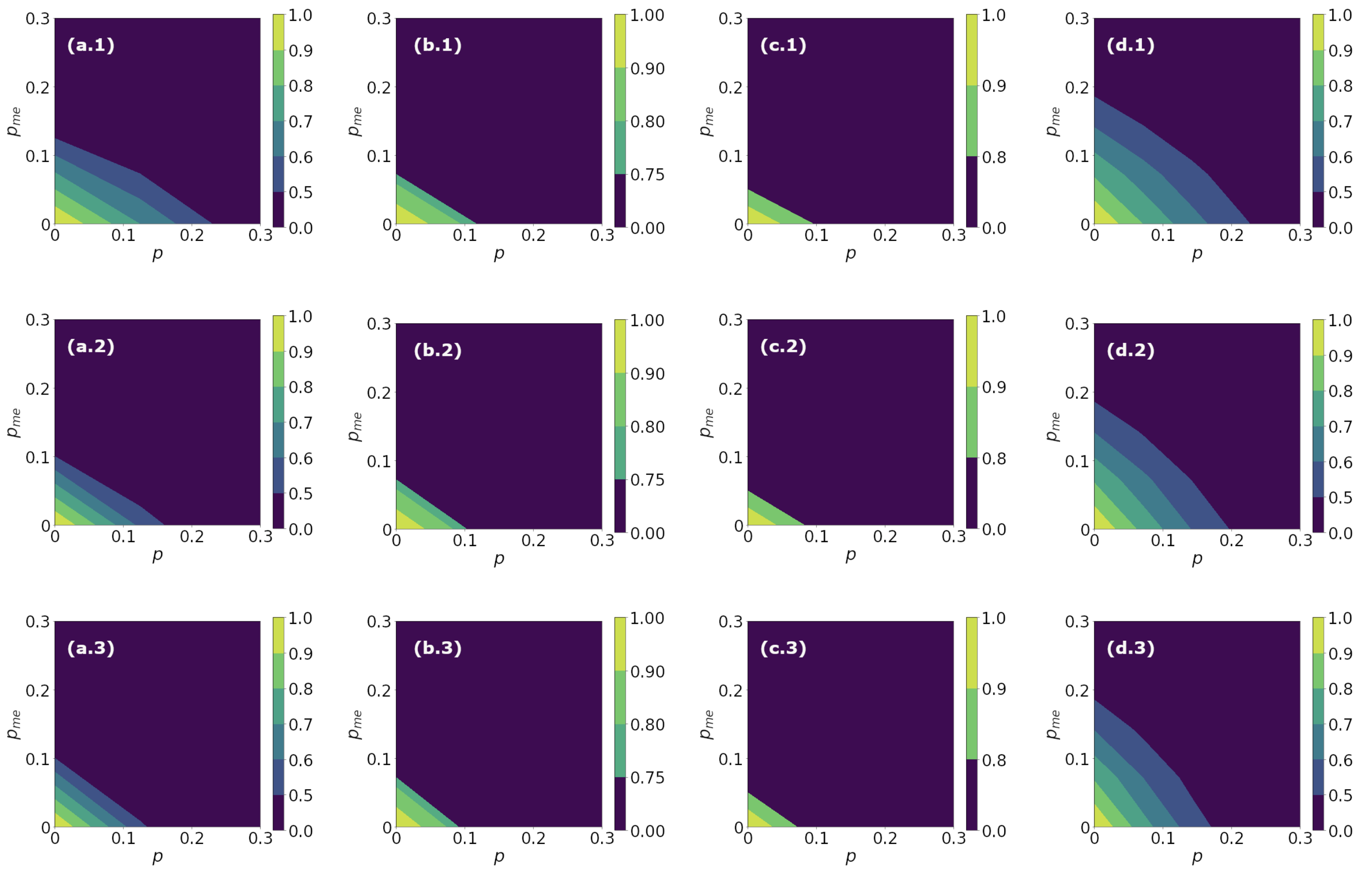}
  \caption{GME witnessing plots
  in a 5-qubit non-FT plaquette readout circuit subject to the two-qubit depolarizing noise model.
  In each row $i$ ($i=1,2,3$) the contour plots represent, for different values of depolarizing-noise probability $p$ and measurement-error probability $p\s{me}$, the expectation values of (a.$i$) the standard GME witness $L\s{GHZ\oprod{5}}$ in Eq.~\eqref{eq:stdWitnessGHZ5_numerical}, (b.$i$) the first SL witness $L\s{TG1}$ in Eq.~\eqref{eq:TG1GHZ5}, (c.$i$) the second SL witness $L\s{TG2}$ in Eq.~\eqref{eq:TG2GHZ5}, and (d.$i$) the worst-case conditional witnessing $L_{[s|x]}$ of Eq.~\eqref{eq:condWitness_5_numeric} for $x=1$. 
  In each row $i=1,2$, and $3$, the error-growth rate per gate values are $r=0,0.1$, and $0.2$, respectively.
  The colormap sidebars represent the witness bounds for each method.
  In all cases, the dark blue region identifies the noise values for which the GME tests are inconclusive.}
  \label{fig:contour_depogates}
\end{figure*}
As mentioned earlier, the depolarizing noise model of Sec.~\ref{subsec:noisemodels} is an oversimplification.
Hence, in this section, we assume a 5-qubit non-FT readout circuit for the plaquette that suffers from the more realistic circuit noise of Eq.~\eqref{eq:noisedepo2}.
As discussed in Sec.~\ref{subsec:noisemodels}, this error model introduces a two-qubit depolarizing channel after each entangling gate. 
Recall that, in this model, the two-qubit gate error probability $p(r, t)$ grows exponentially with the gate time-step $t$ at which $ZZ$-entangling gates are applied 
at the rate $r$. 
Hence, the model takes into account the fact that the fidelity of gates can be affected by the depth of the circuit due to
the accumulated decoherence and noise effects.
As before, we have also included the measurement noise effects as per Eq.~\eqref{eq:noisebitflip}.

Figure~\ref{fig:contour_depogates} represents the resulting contour plots of the expectation values of (a) the standard witness of Eq.~\eqref{eq:stdTestGHZ5-rep},
(b) the SL witness of Eq.~\eqref{eq:TG1GHZ5}, (c) the SL witness of Eq.~\eqref{eq:TG2GHZ5}, and (d) our conditional witness of Eq.~\eqref{eq:condWitness_5_numeric} for $x=1$.
Our choice of the conditional bipartition in Fig.~\ref{fig:contour_depogates}.~d is due to the fact that, here, in contrast to the phenomenological noise, 
the expectation values $\expec{L_{[\syn|x]}}$ are not equal for $x=1,2,3,4$, as shown in the lower panel of Fig.~\ref{fig:thresholds_5q}. 
In this case, the conditional bipartition $[\syn|1]$ is the one in which the conditional GME witnessing performs the worst, providing thus the most stringent conditions for the detection of GME.
Notably, the latter counter-intuitive observation suggests that the depolarizing noise processes before a two-qubit gate are more detrimental than those taking place after it to the correlation created by that gate between the target qubit and all the previously-entangled qubits.
This could be due to the loss of the input coherence that is necessary for correlating the qubits~\cite{PhysRevLett.115.020403}.
As such, qubit $3$, for instance, becomes less correlated with qubits $(\syn,1,2)$ than qubit $2$ with qubits $(\syn,1)$, meaning that conditioning on qubit $3$ is less informative than conditioning on qubit $2$.
Hence, the conditional bipartition $[\syn|3]$ conditioned on the more informative qubits $1,2$ and the less informative qubit $4$ performs better compared to the conditional bipartition $[\syn|2]$ conditioned on the more informative qubits $1$ and the less informative qubits $3,4$.

In Fig.~\ref{fig:contour_depogates} the dark blue shaded area again indicates the states for which each of the witnessing techniques is inconclusive.
We observe that with increasing error-growth rate $r$ the threshold for $p$ decreases.
However, even in the worst-case scenario, our conditional GME witnessing represented in Fig.~\ref{fig:contour_depogates}.~d shows significant robustness compared to the previously-studied entanglement witnesses for each $r$.

It is also interesting to compare the effect of the two different noise models shown in Figs.~\ref{fig:contour_depoattheend} and~\ref{fig:contour_depogates}.
We notice a higher threshold of detectable GME states when the two-qubit depolarizing noise model is applied.
This means that the phenomenological model may underestimate the performance of entanglement witnessing methods whilst the more realistic two-qubit depolarizing noise model predicts a higher robustness of our conditional GME detection.
Nevertheless, it is evident from our analysis here that for both noise models the  GME witnessing by conditioning method introduced in this text 
is not only efficient in the number of qubits, but also robust against noise.\\

\subsection{Conditional GME witnessing in noisy flag-based FT plaquette circuit}\label{subsec:GMEwitness_flag}

We now consider the flag-based FT readout circuit for a single plaquette and implement each noise model of Sec.~\ref{subsec:noisemodels} on the circuit evaluating the different GME witnesses on the output entangled state. 
We use again the three mentioned techniques: (i) standard witnessing with an ideal entangled logical state projector onto the ideal target state~\eqref{eq:Native_FT_6qubit_CohRot} requiring an exponential number of measurements,
(ii) efficient witnessing using the second SL witness method proposed in~\cite{Toth2005} requiring in this case a linear number of measurements, and (iii) conditional GME witnessing requiring a linear number of witnesses and measurements.

As discussed above, the 6-qubit FT flag-based circuit presented in  Fig~\ref{fig:4+2SX_ZZ_circwit} is 
the trapped-ion version of the circuit presented in  Fig.~\ref{fig:ideal_circuits}.~d compiled to the particular set of trapped-ion native gates (see Fig.~\ref{fig:native_gate_scheme}). 
The role of the flag qubit in FT-QEC circuits is to detect cascades of errors from the syndrome qubit to the data qubits after measuring it and extracting the syndrome.
In the majority of cases the flag remains disentangled from the syndrome and therefore from the rest of data qubits. 
However, under some circumstances, such as coherent-rotation syndrome errors of the form $\exp{(i\phi X)}$, the entire plaquette will end up in an entangled state of Eq.~\eqref{eq:Native_FT_6qubit_CohRot}. 

The ideal 6-qubit GME output state of Eq.~\eqref{eq:Native_FT_6qubit_CohRot} can now be used to evaluate the experimental quality of the plaquette measurement circuit. Beginning with the standard witnessing, to show that this state is GME, we have to check the entanglement in all 31 possible bipartitions. Let us focus on one of them, say $({\rm f}|{\rm s},1,2,3,4)$, where we have used the symbol $\rm f$ for the flag qubit and $\rm s$ for the the syndrome qubit. 
It is easy to see that, for instance, setting $\phi={\pi \over 4}$ the output state can be read off as the following bipartite state
\begin{equation}\label{eq:6qCohError}
        \ket{\psi\s{out}}={{\ket{\mathrm{ GHZ}^{\otimes 5}_{+}}\ket{+}\s{f}}+\ket{\widetilde{\mathrm{GHZ}}^{\otimes 5}_{+}}\ket{-}\s{f} \over \sqrt{2}},
\end{equation}
where we denote $\ket{\widetilde{\mathrm{GHZ}}^{\otimes 5}_{+}}=(\ket{00111}-\ket{11000}) /\sqrt{2}$ with $\braket{\mathrm{ GHZ}^{\otimes 5}_{+}}{\widetilde{\mathrm{GHZ}}^{\otimes 5}_{+}}=0$. Consider now the test operator
\begin{equation}\label{eq:stdTestGHZ6}
L=\ket{\psi\s{\rm out}}\bra{\psi\s{\rm out}}.
\end{equation}
For the separability bound with respect to this bipartition, we are effectively dealing with a two-qubit system again which allows us to obtain 
\begin{equation}\label{eq:gsupGHZ6}
    l\s{f|s,1,2,3,4}=\sup_{\sigma\in\sep\s{f|s,1,2,3,4}} \tr({L\sigma}) = {1\over 2}.
\end{equation}
We note, however, that with respect to different bipartitions we may get different bounds.
For each of the 31 bipartitions we have evaluated the bounds and exhaustively listed them in Table.~\ref{app:tab:sepvalues} of Appendix~\ref{app:bounds}.
In each case, we either get a bound of ${1 \over 2}$ or ${1 \over 4}$. 
\begin{figure*}[t!]
  \includegraphics[width=1\textwidth]{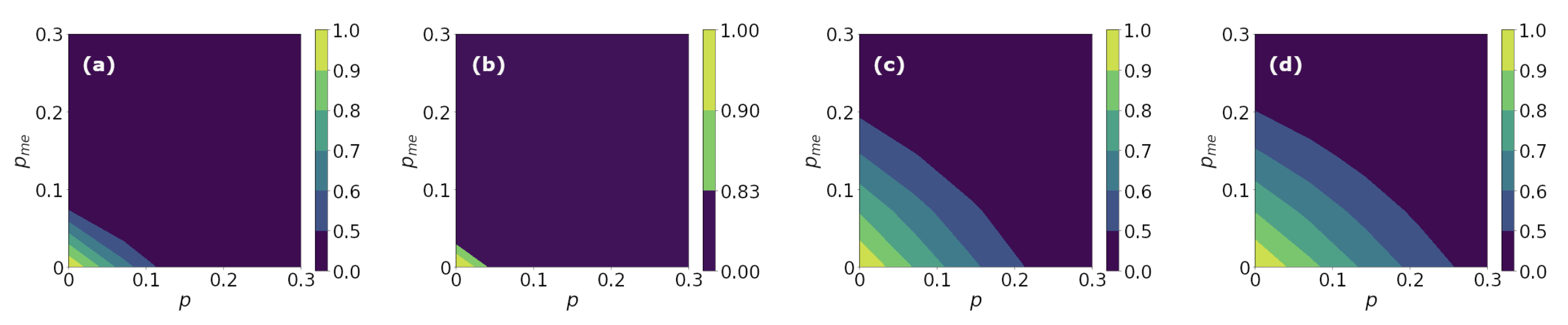}
  \caption{GME witnessing plots
  in a 6-qubit flag-based FT plaquette readout circuit under the phenomenological depolarizing noise model of Sec.~\ref{sec:pheno_6q}.
  The coherent-noise angles is set to $\phi=\pi/4$.
  Each contour plot represents, for different values of depolarizing-noise probability $p$ and measurement-error probability $p\s{me}$, the expectation values of (a) the standard GME witness $L$ of Eq.~\eqref{eq:stdWitnessGHZ6_numerical}, (b) the second SL witness, $L\s{TG2}$ in Eq.~\eqref{eq:TG2GHZ6}, (c) the worst-case conditional witnessing using $L_{[{\rm f}|x]}^{1}$ for $x={\rm s}$. from Eq.~\eqref{eq:condWitness_6_numeric_1}, and (d) the worst-case conditional witnessing using $L_{[{\rm f}|x]}^{2}$ in Eq.~\eqref{eq:condWitness_6_numeric_2} for $x=2$. 
  The colormap sidebars represent the witness bounds for each method.
  In all cases, the dark blue region identifies the noise values for which the GME tests are inconclusive.}
  \label{fig:contour_depoend_6q_pi4}
  \vspace{20pt}
  \includegraphics[width=1\textwidth]{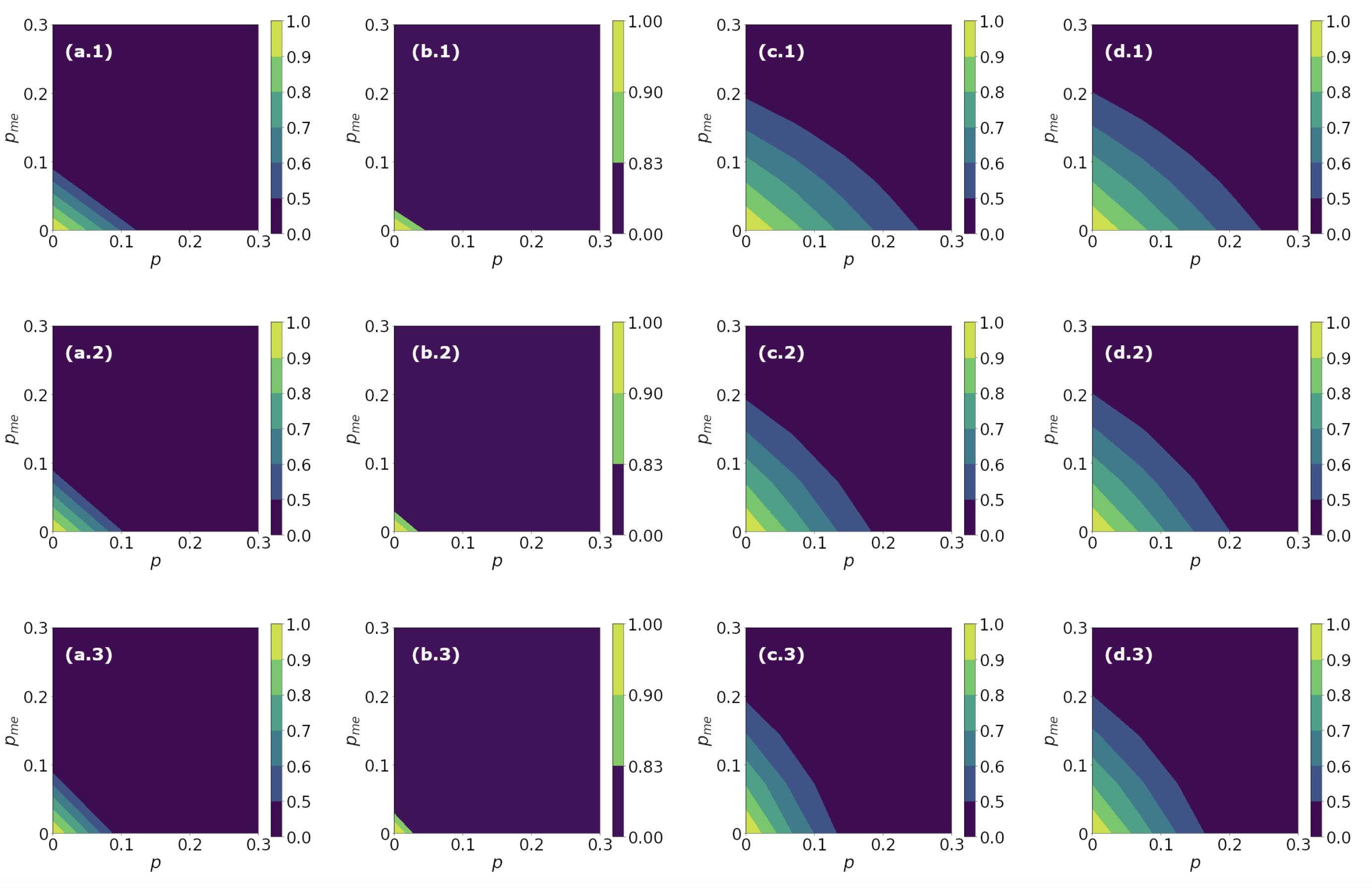}
  \caption{GME witnessing plots
  in a 6-qubit flag-based FT plaquette subject to the two-qubit depolarizing noise model of Sec.~\ref{sec:circ_6q}.
  The coherent-noise angle is set to $\phi=\pi/4$.
  In each row $i$ ($i=1,2,3$) the contour plots represent, for different values of depolarizing-noise probability $p$ and measurement-error probability $p\s{me}$, the expectation values of (a.$i$) $L$ in the standard GME witness $L$ in Eq.~\eqref{eq:stdWitnessGHZ6_numerical}, (b.$i$) the second SL witness
  $L\s{TG2}$ in Eq.~\eqref{eq:TG2GHZ6}, (c.$i$) the worst-case conditional witnessing using $L_{[{\rm f}|x]}^{1}$ for $x={\rm s}$ from Eq.~\eqref{eq:condWitness_6_numeric_1} and (d.$i$) the worst-case conditional witnessing using $L_{[{\rm f}|x]}^{2}$ in Eq.~\eqref{eq:condWitness_6_numeric_2} for $x=2$.  
  In each row $i=1,2$, and $3$, the error-growth rate per gate values are $r=0,0.1$, and $0.2$, respectively.
  The colormap sidebars represent the witness bounds for each method.
  In all cases, the dark blue region identifies the noise values for which the GME tests are inconclusive.}
  \label{fig:contour_depogates_6q}
\end{figure*}
In the standard GME witnessing approach using a single test operator we must choose the largest value (see Eq.~\eqref{eq:GMEmaxBound}) which, in this case, corresponds to ${1 \over 2}$.
Consequently, we can define a GME entanglement witness as
\begin{equation}\label{eq:stdWitnessGHZ6}
    W={1 \over 2}I-L,
\end{equation}
so that $\tr({W\sigma})\geq 0$ for all states that are separable with respect to at least one of the bipartitions, while there exist quantum states for which $\tr({W\varrho})<0$. 
Given any state $\varrho$ an expectation value $\Tr (\varrho L) > {1\over 2}$ thus implies the entanglement of $\varrho$ with respect to all bipartitions, i.e.~its GME.

In Eq.~\eqref{eq:GHZ5Stabilizers}, it was shown that a GHZ-like projector can be written in terms of its $2^n$ stabilizers with $n$ the number of physical qubits.
Following the same procedure for the state in Eq.~\eqref{eq:6qCohError}, the expectation value of the 6-qubit test operator is 
\begin{equation}\label{eq:stdWitnessGHZ6_numerical}
   \expec{L} = \Tr \varrho L = {1 \over 2^{6}}\sum_{i=1}^{2^{6}}\expec{S_i},
\end{equation}
where the stabilizer subgroup is now given by
\begin{equation}\label{eq:6GHZ_stabilizergroup}
\begin{split}
     \mathcal{S}\s{6q} = \langle  g_1& = Z_1X_3X_4X_5Z_6, g_2=Z_2X_3X_4X_5Z_6, \\
     \quad g_3&=Z_3Z_5, g_4=Z_4Z_5,\\
      \quad  g_5&=-X_1X_2Z_5Z_6, g_6=X_3X_4Y_5Y_6 \rangle.
\end{split}
\end{equation}
For the standard witnessing of GME using the 6-qubit projector $L$, it is thus sufficient to locally measure the stabilizers at the output of the plaquette circuit, combine their statistics according to Eq.~\eqref{eq:stdWitnessGHZ6_numerical}, and verify that $\expec{L}>{1 \over 2}$.

The second approach we consider is the SL witnessing of GME~\cite{Toth2005}.
In this occasion, the first SL test operator is not applicable to the 6-qubit state in Eq.~\eqref{eq:6qCohError} due to the lack of required symmetries in stabilizer generators \FS{in} Eq.~\eqref{eq:6GHZ_stabilizergroup}.
The simplest witness we can get is by summing up all the generators as in the second SL witness
\begin{equation}\label{eq:TG2GHZ6}
    L\s{TG2} = {1\over 6}\sum_{i=1}^{6}g_i.
\end{equation}
The GME bound for this witness is $l\s{TG2}={5\over 6}$ ~\cite{Toth2005}.
In contrast to the 5-qubit case, due to the form of the 6-qubit stabilizer generators, this witness would require more than two, yet a linear number of, measurement settings.

Finally, we implement our conditional GME witness. 
Using Theorem~\ref{th:measNo}, we need to test five conditional bipartitions $\rm [f|1]$, $\rm [f|2]$, $\rm [f|3]$, $\rm [f|4]$, and $\rm [f|s]$. In this case, we need to distinguish two different conditional test operators to cover all the 5 conditional bipartitions.
We introduce the first set of conditional test operators as,
\begin{equation}\label{eq:condWitness_6_numeric_1}
    L_{[{\rm f}|x]}^{1}=\ket{\widetilde{\Bell}}_{{\rm f}|x}\bra{\widetilde{\Bell}} \otimes\ket{0}_{1,2}\bra{0}\oprod{2} \otimes \ket{+}_{\{{\rm f},x,1,2\}\us{c}}\bra{+}\oprod{2}
\end{equation}
for $x\in\{3,4,s\}$, in which $\ket{\widetilde{\Bell}}_{{\rm f}|x}\bra{\widetilde{\Bell}}$ is the projection of the qubit pair $f$ and $x$ onto the Bell state $\ket{\widetilde{\Bell}}=(\ket{0+}+\ket{1-})/\sqrt{2}$. 
This conditioning choice is not applicable to $x\in\{1,2\}$ since it would not lead to any Bell state including the flag qubit.
Thus, we use a second set of conditional test operators for the remaining $x\in\{1,2\}$ conditional bipartitions,
\begin{equation}\label{eq:condWitness_6_numeric_2}
    L_{[{\rm f}|x]}^{2}=\ket{\widetilde{\Bell}}_{{\rm f}|x}\bra{\widetilde{\Bell}} \otimes\ket{+}_{\{{\rm f},{\rm s},x,3,4\}\us{c}}\bra{+} \otimes\ket{0}\s{{\rm s},3,4}\bra{0}\oprod{3}.
\end{equation}
As the conditioned two-qubit states are identical for both test operators $L_{[{\rm f}|x]}^{1}$ and $L_{[{\rm f}|x]}^{2}$, the bounds are the same,
\begin{equation}
    l_{{\rm f}|x}^{1,2}=\sup_{\sigma\in\sep_{{\rm f}|x}} \Tr \left(\sigma \ket{\widetilde{\Bell}}_{{\rm f}|x}\bra{\widetilde{\Bell}}\right) = {1 \over 2},
\end{equation}
where the supremum is taken over the set of all bipartite separable states of qubits $\rm f$ and $x$.
Consequently, for any conditional bipartition $[{\rm f}|x]$ with $x\in\{1,2,3,4,s\}$ and any given $6$-qubit state $\varrho$, an expectation value satisfying $\expec{L_{[{\rm f}|x]}^{1,2}}>{1\over 2}$ for all $x\in\{1,2,3,4,s\}$ implies GME of the quantum state $\varrho$.
Importantly for us, the $\ket{\widetilde{\Bell}}$ state is the one-dimensional code space of the stabilizer generators $\{XZ,ZX\}$,
\begin{equation}
    \ket{\widetilde{\Bell}}_{{\rm f}|x}\bra{\widetilde{\Bell}}={I\s{f}I_x+X\s{f}Z_x+Z\s{f}X_x+Y\s{f}Y_x\over 4},
\end{equation}
implying that conditional GME witnessing can be performed solely by local Pauli measurements on qubits.
Furthermore, for the implementation of conditional GME witnessing we need $3$ settings per bipartition, hence a total of $15$ measurement settings.

\begin{figure}[t!]
  \includegraphics[width=0.8\columnwidth]{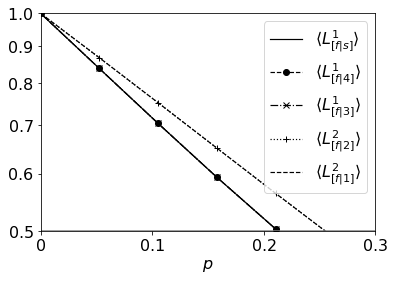}
  \includegraphics[width=0.8\columnwidth]{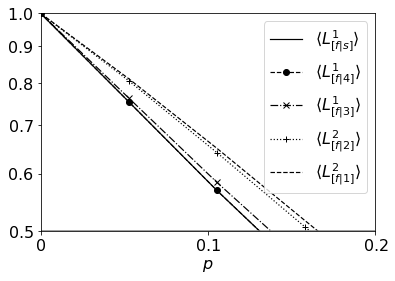}
  \caption{Expectation values of the conditional test operators of Eqs.~\eqref{eq:condWitness_6_numeric_1} and~\eqref{eq:condWitness_6_numeric_2} for five conditional bipartitions for the phenomenological depolarizing noise model (top panel) and for the two-qubit depolarizing noise model (bottom panel) applied to the flag-based FT plaquette.
  The $y$-axes are represented in log-scale whilst the $x$-axes represent the different error probabilities of the two depolarizing error models. 
  In both cases the measurement error probability $p\s{me}$ is assumed to be zero.
  In the bottom panel we have assumed an error-growth rate per gate of $r=0.2$.
  For both noise models, the lowest performance of the test operators $L_{[{\rm f}|x]}^{1}$ and $L_{[{\rm f}|x]}^{2}$ corresponds to $x=s$ and $x=2$, respectively.}
  \label{fig:thresholds_6q}
\end{figure}

\subsubsection{Phenomenological depolarizing noise model}\label{sec:pheno_6q}

The contour plots of Fig.~\ref{fig:contour_depoend_6q_pi4} present the results of the three witnessing techniques with the four test operators to the 6-qubit FT plaquette that suffers from the circuit noise of Eq.~\eqref{eq:noisedepo1} and the measurement noise of Eq.~\eqref{eq:noisebitflip} for a coherent error angle of $\phi=\frac{\pi}{4}$.
From left to right, the four columns of subfigures represent the results obtained for (a) the standard witness of Eq.~\eqref{eq:stdWitnessGHZ6_numerical},
(b) the efficient witness of Eq.~\eqref{eq:TG2GHZ6}, (c) our conditional witness of Eq.~\eqref{eq:condWitness_6_numeric_1} for $x={\rm s}$ and (d) the conditional witness of Eq.~\eqref{eq:condWitness_6_numeric_2} for $x=2$. 
For this error model the expectation values of $L_{[{\rm f}|x]}^{1}$ and $L_{[{\rm f}|x]}^{2}$ behave equally for $x=3,4,{\rm s}$ and $x=1,2$, respectively; see the top panel of Fig.~\ref{fig:thresholds_6q}.
This is due to the symmetry of the error channel with respect to each individual qubit, hence our choices of $x$ in Figs.~\ref{fig:contour_depoend_6q_pi4}.~c and~d.
In all plots the dark blue shaded areas represent states for which tests result values below the separability bounds and no GME is detected.
As we can see, the increase in the noise tolerance of conditional GME witnessing compared to the standard methods is remarkable. 

Setting $\phi=\pi/4$, we find that for $p\s{me}=0$ our conditional witnessing method shows the highest robustness against the depolarizing noise up to $p\approx 0.21$ and $p\approx 0.26$ as shown in Figs.~\ref{fig:contour_depoend_6q_pi4}.~c and~d, respectively.
With the standard method (Fig.~\ref{fig:contour_depoend_6q_pi4}.~a) the noise resilience drops to $p\approx 0.11$ and its lowest value is seen for the second SL method (Fig.~\ref{fig:contour_depoend_6q_pi4}.~b) at $p\approx 0.04$.  
Here again, we clearly observe that linearizing the number of measurements typically results in a reduction of the robustness.
Similarly, by fixing the probability of depolarizing error to $p=0$, one notices that the measurement error probability $p\s{me}$ also reaches its highest threshold value for our conditional witnessing methods both close to $0.2$.
For the standard witnessing method we get a threshold of $p\s{me}\approx 0.08$ whilst it reaches its lowest values for the second SL witness at $p\s{me}\approx 0.04$. 

\subsubsection{Two-qubit depolarizing noise model}\label{sec:circ_6q}

We now assume that the 6-qubit FT plaquette suffers from the more realistic circuit noise of Eq.~\eqref{eq:noisedepo2} which applies 
a two-qubit depolarizing channel after each entangling gate.
As before, we have also included the measurement noise effects as per Eq.~\eqref{eq:noisebitflip}.
Figure~\ref{fig:contour_depogates_6q} represents the resulting contour plots of the expectation values for $\phi={\pi \over 4}$ of (a) the standard witness of Eq.~\eqref{eq:stdWitnessGHZ6_numerical}, (b) the SL witness of Eq.~\eqref{eq:TG2GHZ6}, (c) our conditional witness of Eq.~\eqref{eq:condWitness_6_numeric_1} for $x={\rm s}$ and (d) our conditional witness of Eq.~\eqref{eq:condWitness_6_numeric_2} for $x=2$.
The conditional bipartitions $\rm[f|s]$ and $\rm [f|2]$ chosen here are the ones in which the conditional GME witnessing performs the worst as shown in the bottom panel of Fig.~\ref{fig:thresholds_6q}.
We observe that with increasing error-growth rate $r$ the threshold for $p$ decreases.
However, even for the worst case scenario, the conditional GME witnessing represented in Figs.~\ref{fig:contour_depogates_6q}~c and~d show significant robustness compared to the standard witnesses for each $r$.
By comparing the effects of the two different noise models shown in Figs.~\ref{fig:contour_depoend_6q_pi4} and~\ref{fig:contour_depogates_6q}, one notices again that the more realistic two-qubit depolarizing noise results in higher thresholds of detectable GME states than the phenomenological noise.
Within both error models, the conditional GME witnessing methods are not just more efficient in the number of measurement settings, but also  
more robust against errors compared to their standard counterparts.

\section{ Conclusions and Outlook} \label{sec:conclusions}

We have introduced conditional entanglement witnessing as a robust and efficient technique to test GME in multipartite quantum systems.
We proved that, to show GME in a $n$-partite system it is sufficient to conditionally verify entanglement in just $n-1$ conditional bipartitions.
We then introduced the conditional GME witnessing technique and applied it to non-FT and FT readout circuits for stabilizer operators.
This enabled us to characterize the performance of these circuits as building blocks of larger topological QEC codes in terms of their power for generating GME entanglement.

In particular, we picked a $g_x$ syndrome-readout plaquette of the $d=3$ topological color code and its flag-based FT version adapted to a shuttling-based trapped-ion platform.
We have demonstrated the efficiency and robustness of the conditional witnessing compared to the standard witnessing method when the plaquette circuit undergo phenomenological and two-qubit gate depolarizing noise and measurements suffer from bit-flip errors.   
In terms of efficiency, for the 5-qubit non-FT scheme, in general, the number of measurement settings needed in the standard fidelity witnessing
is 31, whilst the SL approach of Ref.~\cite{Toth2005} needs 2 settings, at the expense of a higher fragility in the presence of noise.
Our conditional witnessing requires 12 settings, but we note that its growth remains linear in the number of qubits in contrast to the exponentially-increasing complexity of the fidelity witness. 

For the 6-qubit FT scheme, the number of measurement settings for the standard fidelity witness rises to 63 whereas our conditional witnessing only requires 15 where one can already observe the exponential versus linear advantage.
In this case, one SL witness worked at the cost of a significant drop in the noise robustness. 
With regards to the robustness of these methods, the conditional GME witness hereby proposed yields a significant increase in the noise tolerance compared to both the standard fidelity estimate, and SL witnesses.

We are convinced that, with the current technology where the addition of each qubit amplifies the unavoidable noise and errors,
it is beneficial to use our conditional entanglement witnessing scheme to certify the entangling power of quantum processors wherein the output states are generic mixed states far from ideal expected output states. 
Moreover, it can also be applied to a broad variety of protocols where efficient entanglement certification plays an important role~\cite{walter2017multipartite,Durt2004,Schaeff2015,Wang2018}.
Last but not least, an interesting direction for further research is to evaluate and compare the performance of the witnessing methods discussed here under more refined microscopic noise models.

\acknowledgements 
We acknowledge fruitful discussions with experimental colleagues J. Hilder, D. Pijn, U. Poschinger and F. Schmidt-Kaler from Johannes Gutenberg Universit\"at Mainz, Germany, as well as with colleagues from the eQual and AQTION collaborations. We gratefully acknowledge support by the EU Quantum Technology Flagship grant AQTION 820495. 
A.R.B acknowledges support by the Universidad Complutense de Madrid-Banco Santander Predoctoral Fellowship.
A.B. acknowledges support from the Ram\'{o}n y Cajal program RYC-2016-20066, and CAM/FEDER Project S2018/TCS- 4342 (QUITEMAD-CM) and the Plan Nacional Generaci\'{o}n de Conocimiento PGC2018-095862-B-C22.
M.M. acknowledges support by the ERC Starting Grant QNets 804247, and also by U.S.A.R.O. through Grant No. W911NF-14-1-010. 
F.S. acknowledges support and resources provided by the Royal Commission for the Exhibition of 1851.

\appendix

\section{$Z$-type parity check circuit with CNOT gates}\label{app:gz_ideal}

Here we provide the $Z$-type parity check circuits in parallel with the $X$-type parity check of Sec.~\ref{sec:ideal_circuits}.
It is assumed that the circuits are ideal and free from state preparation and measurement (SPAM) and gate errors.
We provide the output states from these circuits and their stabilizer generators from which all the entanglement witnesses presented in the text can be obtained. 
  \begin{figure}[t!]
  \includegraphics[width=1\columnwidth]{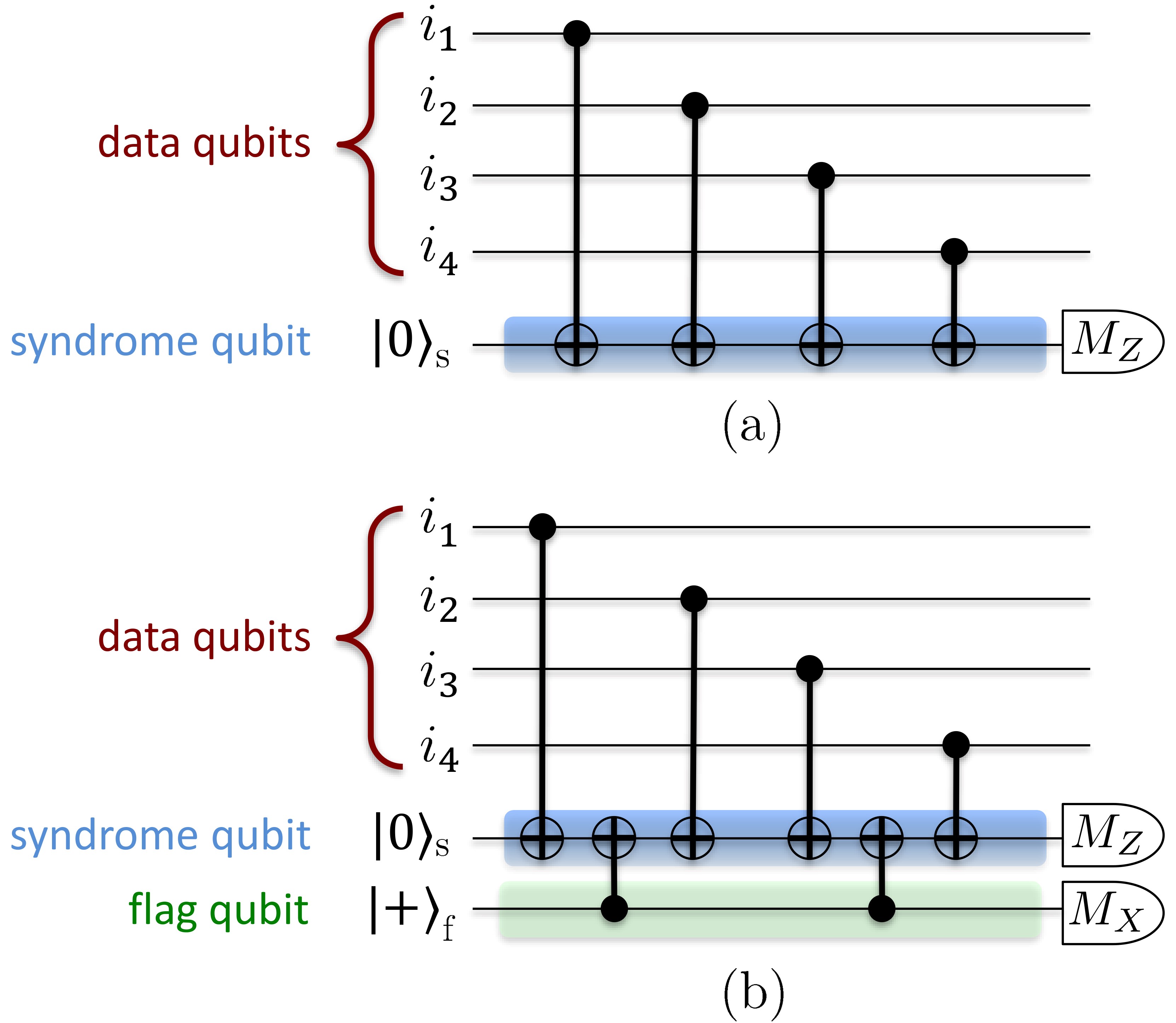}
  \caption{Error syndrome extraction circuits in the color code for (a) the non-FT and (b) the flag-based FT readouts of the generators $g_z^{(q)}$ (Eq.~\eqref{eq:plaquette_gen}) of the plaquette $q$ in Fig.~\ref{fig:color_code}.
  The data qubits $i_1,i_2,i_3,i_4$ correspond to the four physical qubits of this plaquette.}
  \label{fig:stabiliser_readout_sz}
\end{figure}

\subsection{Non-FT circuit}
Starting by the ideal plaquette circuit depicted in Fig.~\ref{fig:stabiliser_readout_sz}.~a representing the non-FT $g_z$ readout circuit,
we initialize the data qubits $\{i_1,i_2,i_3,i_4\}$ in $\ket{+}\oprod{4}$ and the syndrome-qubit in $\ket{0}\s{s}$.
This is expected to produce the output state
\begin{equation}
        \ket{\psi\s{out}}=\ket{\mathrm{GHZ}\oprod{5}_{X+}}={\ket{+}\oprod{5}+\ket{-}\oprod{5} \over \sqrt{2}}.
\label{eq:CNOTGHZ_X_5qubit}
\end{equation}
The resulting $5$-qubit GHZ state (in the $X$-basis) in Eq.~\eqref{eq:CNOTGHZ_X_5qubit} can be understood as the one-dimensional code space of the stabilizer subgroup
\begin{equation}\label{eq:5GHZ_stabilizergroup_cnot}
\begin{split}
     \mathcal{S}_{X,{\rm 5q}} = \langle & g_1 = X_1X_2, g_2=X_2X_3, g_3=X_3X_4,\\
     &\quad  g_4=X_4X_5, g_5=Z_1Z_2Z_3Z_4Z_5\rangle,
\end{split}
\end{equation}

\subsection{Flag-based FT circuit}
We now consider the FT version of the $Z$-type parity check by adding a flag-qubit as shown in Fig.~\ref{fig:stabiliser_readout_sz}.~b.
As in the previous case, the data and the syndrome qubits are initialized in $\ket{+}\oprod{4}$ and $\ket{0}\s{s}$ respectively, whereas the flag qubit is prepared in $\ket{+}\s{f}$.
The expected output state is then given by
\begin{equation}
    \ket{\psi\s{out}}={(\ket{+}\oprod{5} + \ket{-}\oprod{5})\ket{+}\s{f} \over \sqrt{2}}=\ket{\mathrm{GHZ}\oprod{5}_{X+}}\ket{+}\s{f}.
    \label{eq:output_stabilizers_sz}
\end{equation}
As discussed in the main text, in this scheme the flag qubit will detect the presence of a single error on the syndrome-qubit that  propagates to multiple errors on the data qubits, thus enabling the preservation of the encoded information.
For example, if a single Pauli $Z$ error occurs between the second and the third data-to-syndrome CNOT gates in Fig.~\ref{fig:stabiliser_readout_sz}. b,
it leads to the state
\begin{equation}\label{eq:pureout4+2_cz}
\begin{split}
    \ket{\psi\s{out}} = Z\s{s}Z_3Z_4 \ket{\mathrm{GHZ}\oprod{5}_{X+}} \ket{-}\s{f},
\end{split}
\end{equation}
where the flag has been flipped.
A final projective measurement on the flag qubit thus reveals the error.

\section{Exhaustive list of separability bounds for standard witnessing in the flag-based FT circuit}\label{app:bounds}

In Sec.~\ref{subsec:GMEwitness_flag} we used the test operator
\begin{equation}\label{app:eq:stdTestGHZ6}
L=\ket{\psi\s{\rm out}}\bra{\psi\s{\rm out}},
\end{equation}
wherein
\begin{equation}\label{app:eq:6qCohError}
        \ket{\psi\s{out}}={{\ket{\mathrm{ GHZ}^{\otimes 5}_{+}}\ket{+}\s{f}}+\ket{\widetilde{\mathrm{GHZ}}^{\otimes 5}_{+}}\ket{-}\s{f} \over \sqrt{2}}
\end{equation}
with $\ket{\widetilde{\mathrm{GHZ}}^{\otimes 5}_{+}}=(\ket{00111}-\ket{11000}) /\sqrt{2}$.
Here, we give the exhaustive list of separability bounds for $L$,
\begin{equation}\label{app:eq:gsupGHZ6}
    l_B=\sup_{\sigma\in\sep_B} \tr{L\sigma}
\end{equation}
with respect to each bipartition $B$ within the set of all 31 possible bipartitions $\B$ and show that they are either ${1 \over 2}$ or ${1 \over 4}$. 

In the most general form, the solution to Eq.~\eqref{app:eq:gsupGHZ6} is obtained by solving a set of coupled nonlinear eigenvalue equations called \textit{multipartite separability eigenvalue equations}~\cite{Sperling2013,Shahandeh2014}.
This is a highly nontrivial problem to solve in general.
Nevertheless, in some cases, including the example at hand, we can obtain the bounds without directly solving these equations and by merely relying on the properties of $\ket{\psi\s{\rm out}}$.
In particular, the property we use here is the following.
Consider a bipartite pure state of the form
\begin{equation}\label{app:eq:Schmidtform}
\ket{\psi}\s{AB}={1\over\sqrt{n}}\sum_{i=1}^n \ket{i}\s{A}\ket{\psi_i}\s{B},
\end{equation}
such that $\braket{i}{j}=\delta_{ij}$ and $\braket{\psi_i}{\psi_j}=\delta_{ij}$.
Using techniques from Refs.~\cite{Sperling2013,Shahandeh2014}, we infer that the maximum separability bound is 
\begin{equation}
    l\s{AB}=\sup_{\sigma\in\sep\s{AB}} \tr{\ketbra{\psi}\sigma} = {1 \over n}.
\end{equation}
In other words, biorthogonality properties of the local Schmidt bases allow us to read the separability values off the Schmidt rank of $\ket{\psi}$.
It turns out that that we can use this result in the specific case of $\ket{\psi\s{out}}$ in Eq.~\eqref{app:eq:6qCohError}.
We skip writing down Schmidt decompositions with respect to individual bipartitions, as this is straightforward, and merely list the separability bounds in Tab.~\ref{app:tab:sepvalues}. 
Nevertheless, to give an example, we notice that $\ket{\psi}\s{out}$ in Eq.~\eqref{app:eq:6qCohError} is written in its Schmidt decomposition with respect to the bipartition $\rm (f|s,1,2,3,4)$ with an Schmidt rank of $2$ and it is of the form of Eq.~\eqref{app:eq:Schmidtform}, hence $l\s{f|s,1,2,3,4}={1\over 2}$.

\begin{table}[t!]
\begin{tabular*}{\columnwidth}{c @{\extracolsep{\fill}} c}
\hline
\hline
Bipartition $B$ & \hfill $l_B$ \\
\hline
$\rm (f|s,1,2,3,4)$ & ${1 / 2}$ \\
$\rm (s|f,1,2,3,4)$ & ${1 / 2}$ \\
$\rm (1|f,s,2,3,4)$ & ${1 / 2}$ \\
$\rm (2|f,s,1,3,4)$ & ${1 / 2}$ \\
$\rm (3|f,s,1,2,4)$ & ${1 / 2}$ \\
$\rm (4|f,s,1,2,3)$ & ${1 / 2}$ \\
$\rm (f,s|1,2,3,4)$ & ${1 / 4}$ \\
$\rm (f,1|s,2,3,4)$ & ${1 / 4}$ \\
$\rm (f,2|s,1,3,4)$ & ${1 / 4}$ \\
$\rm (f,3|s,1,2,4)$ & ${1 / 4}$ \\
$\rm (f,4|s,1,2,3)$ & ${1 / 4}$ \\
$\rm (s,1|f,2,3,4)$ & ${1 / 4}$ \\
$\rm (s,2|f,1,3,4)$ & ${1 / 4}$ \\
$\rm (s,3|f,1,2,4)$ & ${1 / 2}$ \\
$\rm (s,4|f,1,2,3)$ & ${1 / 2}$ \\
$\rm (1,2|f,s,3,4)$ & ${1 / 2}$ \\
$\rm (1,3|f,s,2,4)$ & ${1 / 4}$ \\
$\rm (1,4|f,s,2,3)$ & ${1 / 4}$ \\
$\rm (2,3|f,s,1,4)$ & ${1 / 4}$ \\
$\rm (2,4|f,s,1,3)$ & ${1 / 4}$ \\
$\rm (3,4|f,s,1,2)$ & ${1 / 2}$ \\
$\rm (f,s,1|2,3,4)$ & ${1 / 4}$ \\
$\rm (f,s,2|1,3,4)$ & ${1 / 4}$ \\
$\rm (f,s,3|1,2,4)$ & ${1 / 4}$ \\
$\rm (f,s,4|1,2,3)$ & ${1 / 4}$ \\
$\rm (f,1,2|s,3,4)$ & ${1 / 2}$ \\
$\rm (f,1,3|s,2,4)$ & ${1 / 4}$ \\
$\rm (f,1,4|s,2,3)$ & ${1 / 4}$ \\
$\rm (f,2,3|s,1,4)$ & ${1 / 4}$ \\
$\rm (f,2,4|s,1,3)$ & ${1 / 4}$ \\
$\rm (f,3,4|s,1,2)$ & ${1 / 4}$ \\
\hline
\hline
\end{tabular*}
\caption{Separability values for the test operator of Eq.~\eqref{eq:stdTestGHZ6} with respect to all 31 bipartitions.}
\label{app:tab:sepvalues}
\end{table}
\section{$Z$-type parity check circuits using native trapped-ion gates}\label{app:gz_ideal_ZZ}
We now translate the $Z$-type parity-check circuits of Appendix~\ref{app:gz_ideal}
into the trapped ions native gates suing the relations given in Fig.~\ref{fig:native_gate_scheme}.
We also provide the output states from these circuits and their stabilizer generators by means of which the entanglement witnesses given in the main text can be constructed.

\subsection{Non-FT circuit}

\begin{figure}[t]
\centering
\includegraphics[width=0.8\columnwidth]{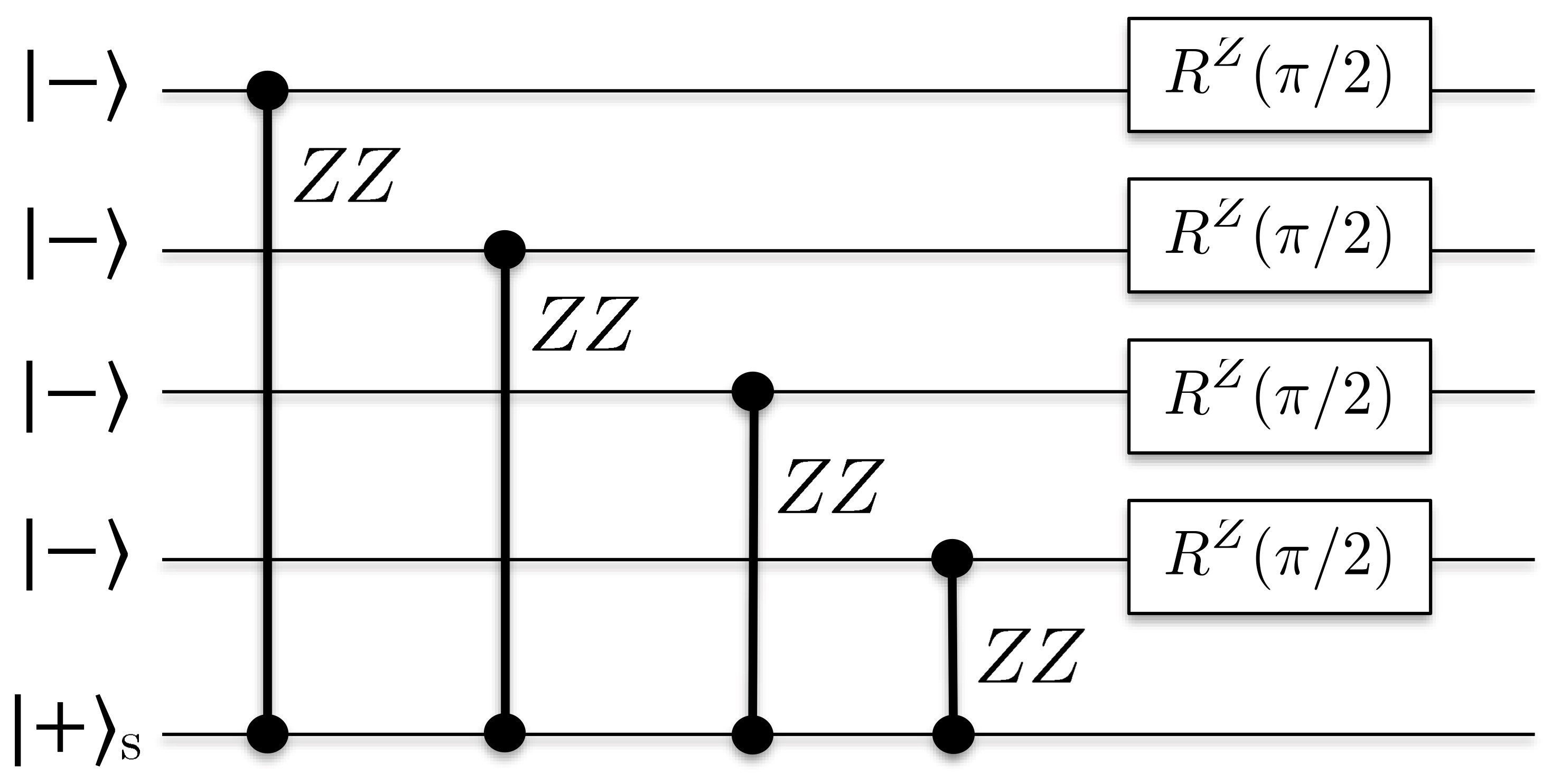}
\caption{The non-FT circuit using the native trapped-ion gates.}
\label{SZ_zz_noerrors}
\includegraphics[width=0.8\columnwidth]{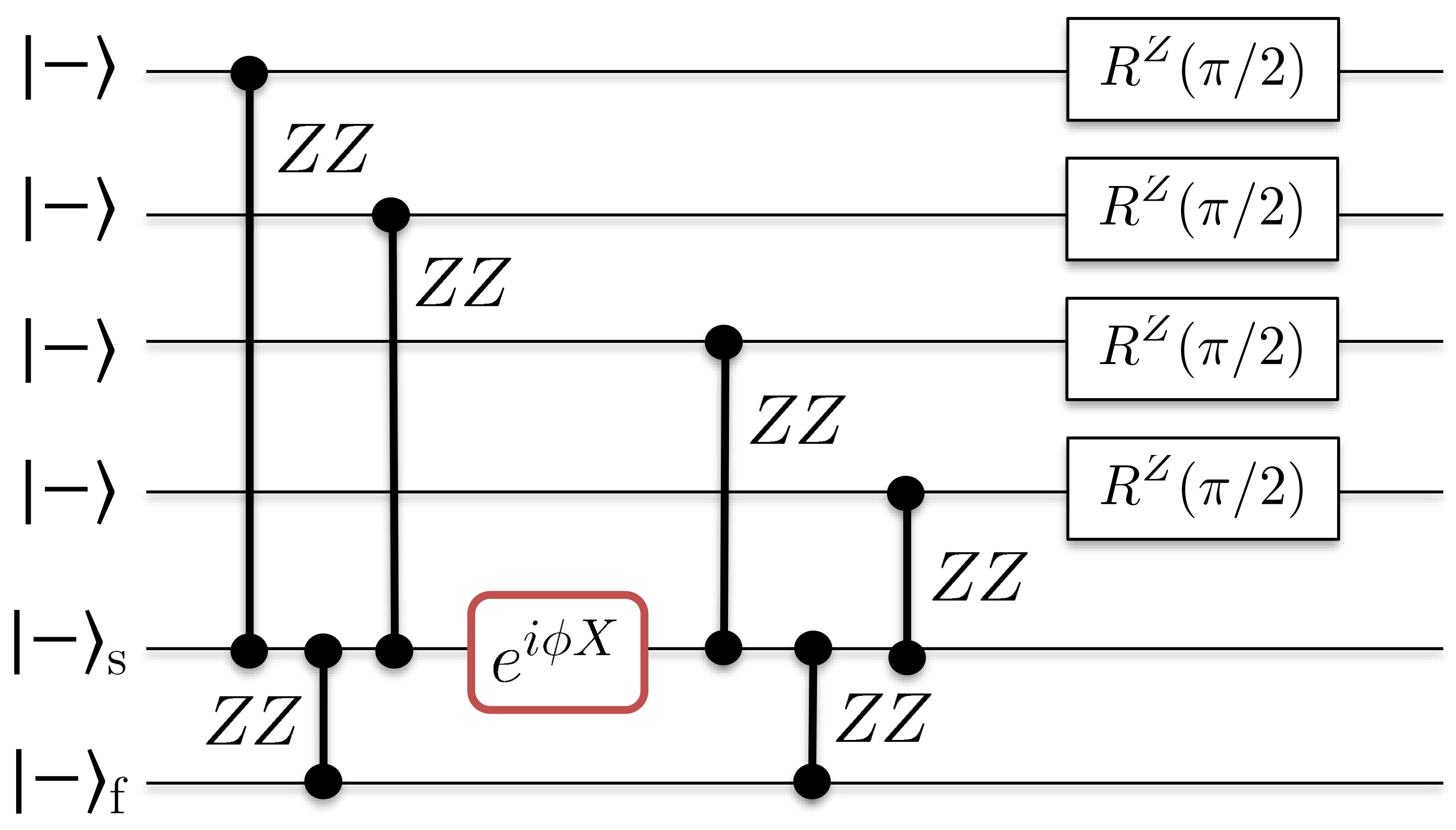}
\caption{ The flag-based FT circuit using the native trapped-ion gates.
    The $e^{i\phi X}$ error is injected to create a 6-qubit entangled state in the flag-based FT stabilizer measurement.}
\label{SZ_flag5_12}
\end{figure}
In Fig.~\ref{SZ_zz_noerrors} the circuit for $Z$-type syndrome readout using the native trapped-ion language is shown.
We start by initializing the data qubits in $\ket{-}\oprod{4}$ and the syndrome-qubit in $\ket{+}\s{s}$.
Similar to the CNOT-based circuit, this is expected to produced the output state
\begin{equation}
        \ket{\psi\s{out}}=\ket{\mathrm{GHZ}\oprod{5}_{X+}}={\ket{+}\oprod{5}+\ket{-}\oprod{5} \over \sqrt{2}},
\label{eq:ZZGHZ_X_5qubit}
\end{equation}
with the same stabilizer generators as given in Eq.~\eqref{eq:5GHZ_stabilizergroup_cnot}.
This means that $\ket{\psi\s{out}}$ in Eq.~\eqref{eq:ZZGHZ_X_5qubit} is a $+1$ eigenstate of the five generating stabilizers $g_i$ from Eq.~\eqref{eq:5GHZ_stabilizergroup_cnot} and all their possible combinations. 
\subsection{Flag-based FT circuit}
Here, we present the flag-based $g_z$ readout circuit using the trapped-ion native gate set. 
Recall that, to generate entanglement between the flag, syndrome, and data qubits we need to inject an error of the form exp($i \phi X$) in a controllable way; see Fig.~\ref{SZ_flag5_12}.
In this occasion we initialize the data qubits in $\ket{-}\oprod{4}$, the syndrome in $\ket{-}\s{s}$, and the flag qubit in $\ket{-}\s{f}$.
The output state is then
\begin{equation}\label{eq:6qubit_GME}
\begin{split}
    \ket{\psi\s{out}} = & \cos{\phi}\ket{+}\s{f}\ket{\mathrm{GHZ}\oprod{5}_{X+}} + \\
    & i\sin{\phi} \ket{-}\s{f} Y\s{s}Z_{3}Z_{4} \ket{\mathrm{GHZ}\oprod{5}_{X+}},
\end{split}
\end{equation}
where the stabilizer subgroup is given by
\begin{equation}\label{eq:6GHZ_stabilizergroup_zz}
\begin{split}
     \mathcal{S}\s{\tilde{X},6q} = \langle & g_1 = X_1Z_3Z_4Z_5Z_6, g_2=X_2Z_3Z_4Z_5Z_6, \\
     &\quad g_3=X_3X_5, g_4=X_4X_5,\\
      &\quad  g_5=-Z_1Z_2X_5Z_6, g_6=-Z_3Z_4Y_5Y_6 \rangle.
\end{split}
\end{equation}

\section{Error propagation through $ZZ$-gates}\label{app:ZZerror_propagate}
In Fig.~\ref{fig:error_ZZ} we give the Pauli-error propagation rules through the $ZZ$-entangling gates of trapped-ion platforms.
These come in handy in calculating the output states, e.g.~that of Eq.~\eqref{eq:Native_FT_6qubit_CohRot}.
\begin{figure}[h]
    \centering
    \includegraphics[width=0.8\columnwidth]{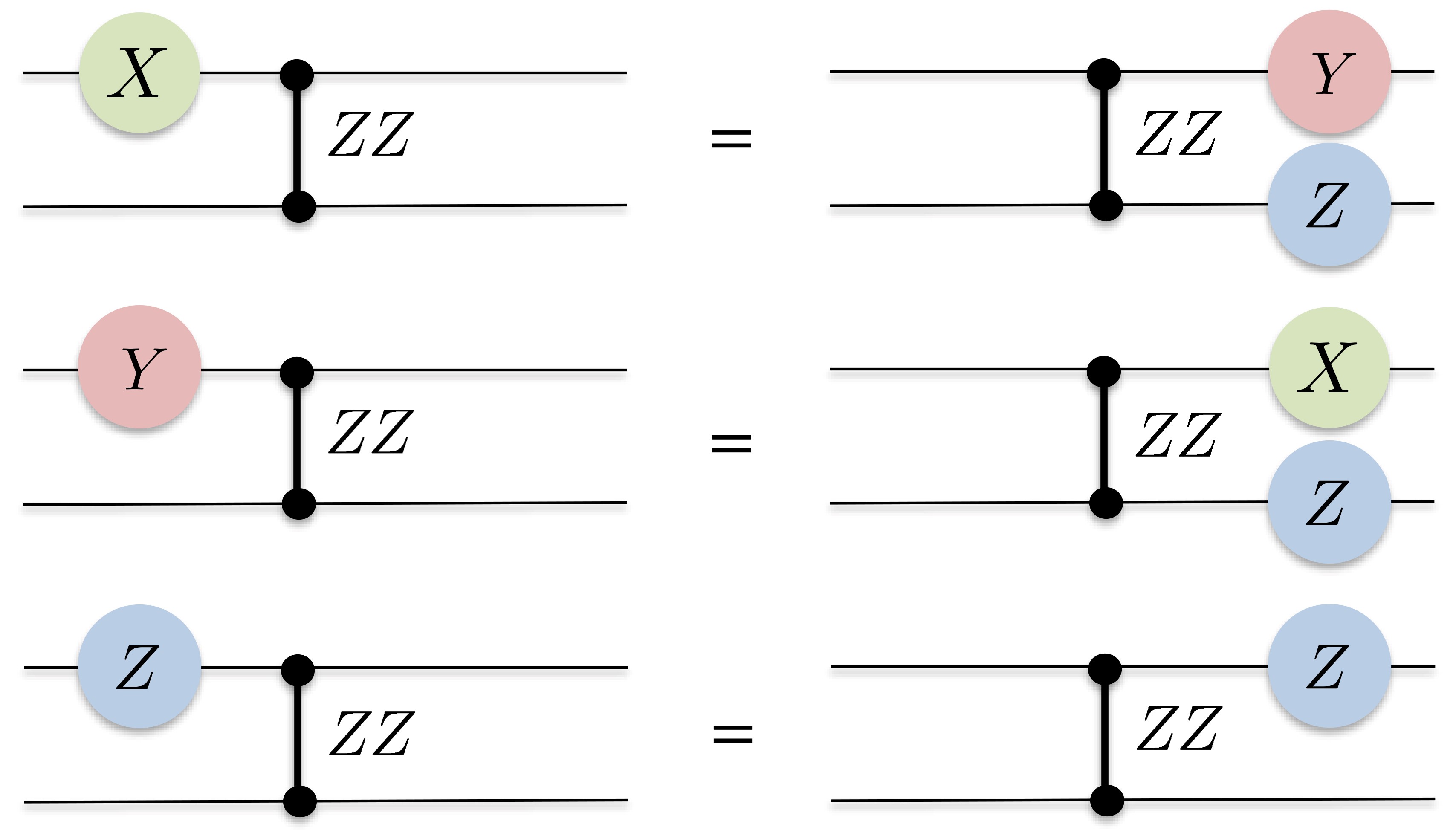}
    \caption{Error propagation through $ZZ$-gates. A Pauli-$X$ ($Y$) error propagates as a $Y$ ($X$) error and a $Z$ error onto the second qubit. 
    An incoming single $Z$ error commutes with the $ZZ$-gate and therefore does not propagate onto the second qubit.}
    \label{fig:error_ZZ}
\end{figure} 

\bibliography{EntInQEC}
\bibliographystyle{apsrev4-1new}

\end{document}